\documentclass[journal]{IEEEtran}
\usepackage{ifpdf}
\ifpdf
\else
\fi

\usepackage[nocompress]{cite}

\ifCLASSINFOpdf
\usepackage[pdftex]{graphicx}
\else
\usepackage[dvips]{graphicx}
\fi
\usepackage[T1]{fontenc}
\usepackage{epstopdf}
\usepackage{overpic}
\usepackage{amsmath,amssymb,amsfonts,latexsym}
\usepackage{upgreek}

\newtheorem{proposition}{\bf Proposition}

\def\tr{\mathrm{Tr}}
\def\re{\mathrm{Re}}

\allowdisplaybreaks[4]

\usepackage[linesnumbered, ruled]{algorithm2e}
\usepackage{algorithmic}  

\usepackage{float}

\usepackage{array}
\usepackage{multirow,tabularx}
\usepackage{textcomp}
\usepackage{stfloats}
\usepackage{url}
\usepackage{xurl}
\usepackage{verbatim}
\usepackage[caption=false]{subfig}
\usepackage{enumitem}
\usepackage{balance}
\usepackage[table]{xcolor}
\usepackage{booktabs}
\usepackage{pifont}

\usepackage{bbding}
\usepackage{fixltx2e}
\usepackage{color}
\usepackage{comment}
\usepackage{bm}
\usepackage{hyperref}

\hypersetup{
	pdftitle={Movable Antenna for Anti-Jamming Transmission: Continuous and Discrete Positioning Designs},
	pdfauthor={Yifan Guo, Junshan Luo, Shilian Wang, Zhenhai Xu, and Ana García Armada},
	bookmarks=true,
	bookmarksopen=true,
	bookmarksnumbered=true,
	colorlinks=true,
	linkcolor=blue,
	citecolor=blue,
	urlcolor=magenta
}

\begin{document}
\title{Movable Antenna for Anti-Jamming Transmission: Continuous and Discrete Positioning Designs}
\author{{Yifan~Guo,~Junshan~Luo,~Shilian~Wang,~Zhenhai~Xu,~and~Ana García Armada,~\IEEEmembership{Fellow, IEEE}}
	\thanks{Yifan Guo, Junshan Luo, Shilian Wang and Zhenhai~Xu are with College of Electronic Science and Technology, National University of Defense Technology, Changsha 410073, China (e-mails: guoyifan@nudt.edu.cn, ljsnudt@foxmail.com, wangsl@nudt.edu.cn, drxzh930@sina.com).}
	\thanks{Ana García Armada is with the Department of Signal Theory and Communications, Universidad Carlos III de Madrid, 28911 Leganés, Spain (e-mail:anagar@ing.uc3m.es).}
}

\IEEEaftertitletext{\vspace{-0.8cm}}
\maketitle

\begin{abstract}
	Movable antenna (MA) system holds big promise in further wireless communications due to the ability of reconfiguring the antenna position to enhance the spatial diversity. This paper investigates the joint design of beamforming and antenna positions in MA-assisted anti-jamming communications. We consider a multi‑user multiple-input multiple-output downlink scenario where multiple malicious jammers exist. The goal is to maximize the sum-rate under the imperfect jammer channel state information. Then, we develop two distinct optimization frameworks for continuous and discrete antenna position designs, respectively. For the continuous position design, we propose an alternating optimization (AO) framework that integrates successive convex approximation and majorization minimization (MM) to handle the highly non-convex problem. For the discrete position design, based on the minimum mean squared error criterion and MM, we reformulate the problem as a sparse recovery task and propose a low-complexity heuristic block coordinate descent and simultaneous orthogonal matching pursuit, which enables joint design rather than AO. Numerical results show that the proposed discrete-position heuristic achieves higher sum rates than the adopted AO-based continuous-position algorithm. This observation highlights the trade-off between spatial resolution and optimization tractability, indicating that the sparse-recovery formulation can yield better performance than the continuous AO algorithm in the MA-assisted anti-jamming scenarios.
\end{abstract}


\section{Introduction}
\IEEEPARstart{T}{he} evolution of beyond fifth-generation (B5G) and forthcoming sixth-generation (6G) wireless networks has aroused tremendous demands for higher data rates, ultra reliability, and massive device connectivity \cite{10379539}. However, the inherent broadcast and superposition nature of open wireless media makes wireless networks more vulnerable to jamming attacks, leading to the potential failure of communication. To cope with this issue, some traditional schemes, e.g., frequency hopping \cite{9733393} and direct-sequence spread spectrum \cite{5472426} have been extensively studied and developed. Nevertheless, these techniques consume extra spectral resources and may become ineffective in the presence of full-band jamming. Therefore, massive multiple-input multiple-output (MIMO) techniques are adopted to enhance array aperture gain and degrees-of-freedom (DoFs) for effective jamming mitigation \cite{8017521}. However, conventional MIMO systems typically assume fixed-position antenna (FPA) configurations, which {cannot} fully exploit spatial DoFs and adapt to wireless channel variations.

Recently, movable antenna (MA) \cite{MA1, MA2, MA-Sur}, also referred to as fluid antenna systems (FASs) \cite{PL-FAS,FAS_Wong,FAMA}, have been proposed to overcome the limitations of FPAs. By dynamically adjusting the antenna position within a confined spatial region, MA can achieve spatial reconfigurability to create additional spatial DoFs, thereby improving the tradeoff between desired signal strengthening and jamming suppression \cite{CL}. This capability enables MA to significantly improve anti-jamming performance, providing a robust framework for addressing the reliability and security requirements of future communication systems.
\subsection{Related Works}
MA-aided wireless communications exploit the spatial variability of wireless channels through antenna position optimization within a prescribed region \cite{MA1, MA2}. Depending on the feasible set of antenna positions, antenna position design can be formulated using either continuous or discrete positioning. In continuous-position designs, antenna positions are optimized over a given continuous region \cite{MA1, MA2}. In discrete-position designs, antenna positions are selected from a predefined set of candidate locations, typically arranged on a grid. Such discrete-position models have also been widely investigated in related studies on MA, including the early work by Wong et al. \cite{FAS, FAS_Wong}.
\par
Recently, research on discrete-position reconfiguration, particularly in FASs, has attracted considerable attention, covering performance analysis \cite{PL-FAS,FAS,FAMA,FASOPDG,MIMO-FAS,SFAMA, 11480436}, index modulation (IM) design \cite{IM-FAS,PIM-FAS}, and antenna position design or port selection \cite{PS1, PS2, BnBFAS}. Specifically, the diversity gains and outage performance of FASs have been investigated in \cite{FASOPDG, MIMO-FAS}, while their multi-user (MU) multiplexing performance has been analyzed in \cite{FAMA, SFAMA}. Variable block-correlation modeling has also been applied to characterize FAS secrecy performance \cite{11480436}. By conveying additional information through the indices of discrete antenna positions or position patterns, IM has been integrated with FASs to enhance spectral efficiency \cite{IM-FAS, PIM-FAS}. For port selection, deep-learning-based approaches have been developed to identify suitable ports from observations at a limited number of ports \cite{PS1, PS2}. A branch-and-bound (BnB)-based method has also been proposed for MA port selection \cite{BnBFAS} to minimize transmit power. Beyond port selection, flexible weighted minimum mean-square error (F-WMMSE) beamforming has been proposed to jointly optimize beamforming and discrete antenna positions in MA-aided MU-MIMO communications \cite{F-WMMSE}. Beyond communication applications, scalable FASs have exploited aperture reconfiguration for source localization \cite{11434508}.
\par
Alternatively, the implementation of motor-driven MA was later introduced in \cite{MA1, MA2}, enabling continuous position optimization and maximizing the potential of spatial diversity. In \cite{MA2}, a field-response channel model was proposed, and the maximum achievable channel gain of the MA was analytically characterized. Building upon this, further studies have demonstrated that, by exploiting continuous antenna position optimization, significant performance gains over FPA systems can be achieved across a variety of performance objectives, including achievable rates \cite{WSRMA, JSACMA}, physical layer security (PLS) \cite{SAB, RIS-MA, NF-MA}, sensing Cramér-Rao bound (CRB) \cite{FAS-Sensing, FAS-ISAC}, etc. Specifically, to optimize the system’s weighted sum achievable rate, WMMSE algorithm and majorization-maximization (MM) algorithm were used \cite{WSRMA}. In \cite{JSACMA}, employing fractional programming, the authors broke down joint beamforming and continuous antenna position optimization into four separate blocks and merged them with the block coordinate ascent (BCA) method for obtaining a stationary solution. To maximize the multi-cast secrecy rate, penalty constrained product manifold framework was employed to address the non-convexity of the problem and obtain high-quality sub-optimal solutions \cite{SAB}. Furthermore, MA has been integrated into emerging systems to improve PLS, such as integrated sensing and communication (ISAC) \cite{RIS-MA} and near-field communication \cite{NF-MA}. Furthermore, the continuous antenna position optimization has validated the effectiveness in improving sensing performance. In \cite{FAS-Sensing}, the authors considered a CRB minimization problem for target angle estimation in both one-dimensional (1D) and two-dimensional (2D) continuous MA. Then, \cite{FAS-ISAC} evaluated the sensing performance gain of MA in ISAC by solving the CRB minimization problem using a boundary traversal search algorithm.
\vspace{-5mm}
\subsection{Motivations}
The existing MA literature generally holds that the continuous design is inherently superior to the discrete design \cite{CFAS-NOMA, 11419085}. This view is based on the fact that the continuous domain provides an infinite space of feasible solutions, theoretically enabling more precise beam alignment or nulling. However, this intuitive conclusion inevitably overlooks a crucial trade-off that the contradiction between spatial flexibility and the tractability of algorithmic design. Specifically, the continuous design faces the challenge of high non-convexity when optimizing antenna positions, which relies on alternating optimization (AO) strategies, making the algorithms easy to get trapped in poor local optima \cite{WSRMA, JSACMA, RIS-MA, NF-MA}. {Moreover, discrete positioning is more practical, as antenna positions are selected from a finite set of available antenna ports.}

Motivated by these considerations, this paper investigates the trade-off between spatial resolution and optimization tractability in antenna position optimization. To this end, we compare the discrete joint-design algorithm with the continuous AO algorithm in terms of anti-jamming performance and computational complexity. By exploiting grid-induced sparsity, the discrete design enables joint antenna selection and beamforming design through a sparse recovery formulation \cite{F-WMMSE}. However, no literature has systematically investigated the two paradigms in the contexts of MA-aided anti‑jamming communications, especially under imperfect jammer channel state information (CSI). This paper is driven by the following three specific questions:
\begin{itemize}
	\item[$\blacktriangleright$] \textit{Which optimization framework can achieve better anti-jamming performance?} Continuous design offers theoretical infinite spatial resolution \cite{CFAS-NOMA,11419085}, but AO may converge to poor local optima in complex anti-jamming problems. The discrete position design, in contrast, enables joint optimization through sparse recovery but is limited by grid resolution. A systematic comparison under identical conditions is needed to determine the specific scenarios and constraints where each paradigm exhibits its practical advantages.
	
	\item[$\blacktriangleright$] \textit{How does each design behave under extreme conditions?} Anti-jamming communication must withstand high jamming power and uncertain jammer CSI. While continuous design is expected to be more flexible \cite{ICCAntiJamming,CL}, discrete selection may be more vulnerable to estimation errors due to its finite grid. Quantifying their robustness in challenging scenarios is essential for practical deployment.
	
	\item[$\blacktriangleright$] \textit{What accounts for the anti-jamming performance differences between the continuous- and discrete-position algorithms?} The continuous- and discrete-position frameworks differ in their feasible position sets and update strategies. Therefore, an analysis of the factors underlying these performance differences helps clarify the trade-offs between spatial resolution and optimization tractability, as well as their implications for achievable anti-jamming performance.
\end{itemize}
\vspace{-3mm}
\subsection{Contributions}
Motivated by these observations, we study when the proposed discrete-position algorithm can obtain a better feasible solution than the considered continuous-position AO algorithm, and at what computational cost. The main contributions are summarized as follows:
\begin{itemize}
	\item \textit{Framework:} In this paper, we establish a MA‑assisted MU‑MIMO downlink network model that explicitly accounts for angular uncertainty in jammer CSI. The proposed framework maximizes the sum rate under maximum transmit-power, antenna-movement, and minimum-rate constraints.
	
	\item \textit{Continuous Position Design:} For continuously adjustable antenna positions, we develop an AO algorithm combining successive convex approximation (SCA), majorization minimization (MM), and minimum mean-square-error (MMSE) reception. The method monotonically improves its surrogate objective and converges to a stationary solution under the standard SCA conditions.
	
	\item \textit{Discrete Position Design:} For the discrete position and beamformer design, we first discretize the continuous antenna moving region and reformulate the joint optimization problem as a sparse recovery task based on MMSE criterion and MM. The key point is that this sparsity modeling reveals the potential for joint optimization of beamforming and antenna position. In particular, we employ a regularized least squares-based simultaneous orthogonal matching pursuit (RLS-SOMP) method, which leverages the inherent sparsity of the antenna activations to directly and jointly identify candidate position indices and the associated beamforming coefficients in a greedy manner.
	
	\item \textit{Analysis:} We perform a series of performance analysis discussions to assess the proposed framework from three perspectives: {sum-rate performance}, {beam alignment capability}, and {robustness}. Numerical results demonstrate that our proposed discrete design algorithm outperforms the state-of-the-art benchmark by 13.3\% and surpasses the AO-based continuous design by approximately 22.4\% under high jamming power conditions. Moreover, both proposed schemes exhibit high robustness, maintaining superior performance over existing alternatives even in the presence of high uncertainty in jammer CSI and high-power jamming attacks.
\end{itemize}

{The remainder of this paper is organized as follows. The system model and the problem formulation are presented in Section \ref{sec:sys}. In Sections \ref{sec:conti} and \ref{sec:discre}, we develop the continuous- and discrete-position designs, respectively. Numerical results are provided in Section \ref{sec:sim}. We conclude this paper in Section \ref{sec:conl}.}
\begin{figure*}[!t]
	\centering
	\includegraphics[width=0.7\textwidth]{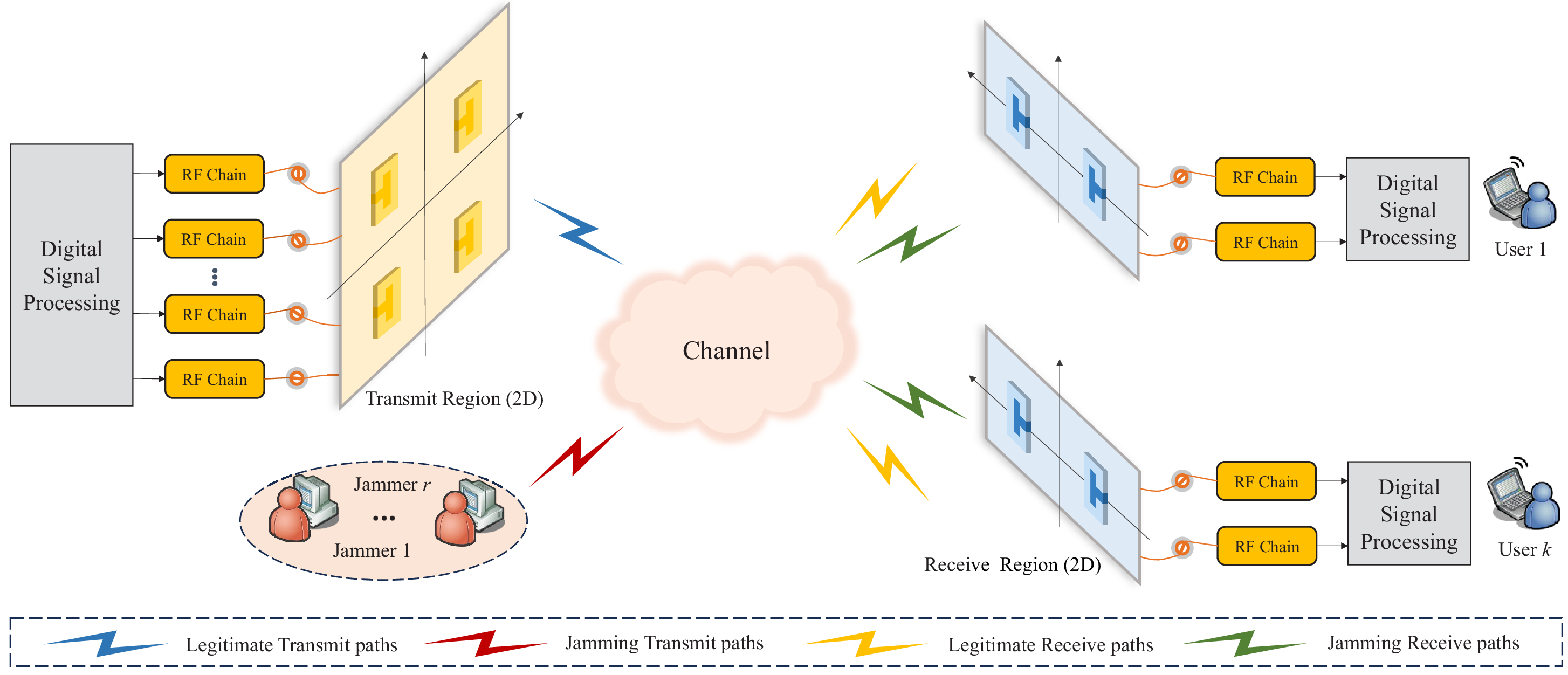}
	\caption{Downlink anti-jamming transmission of MA-assisted MU-MIMO networks.}
	\label{fig:sysmodel}
	\vspace{-5mm}
\end{figure*}

\textit{Notations:} $a$, $\mathbf{a}$, $\mathbf{A}$, and $\mathcal{A}$ represent a scalar, a vector, a matrix, and a set, respectively. $(\cdot)^{\mathrm{T}}$, $(\cdot)^{*}$, $(\cdot)^{\mathrm{H}}$, and $(\cdot)^{-1}$ denote transpose, conjugate, conjugate transpose, and inverse, respectively. The Euclidean norm of a vector $\mathbf{a}$ is expressed as $\Vert \mathbf{a} \Vert$. $\mathrm{Re}(a)$ and $\vert a \vert$ denote the real part and the absolute value of a scalar $a$, respectively. For a matrix $\mathbf{A}$, $\tr(\mathbf{A})$ and $\Vert \mathbf{A} \Vert_{F}$ denotes its trace and the Frobenius norm. Matrix entries are indexed by $[\mathbf{A}]_{m,n}$, corresponding to the element at the $m$-th row and $n$-th column. $\mathbf{I}_{N}$ denotes an identity matrix of dimension $N \times N$. $\mathbf{0}_{M}$ denotes the $M$-dimensional all-zero vector. $\mathcal{CN}(\mathbf{0},\mathbf{\Lambda})$ denotes the circularly symmetric complex Gaussian (CSCG) distribution with mean zero and covariance matrix $\mathbf{\Lambda}$. $\mathbb{R}$ and $\mathbb{C}$ represent the sets of real and complex numbers, respectively. $\mathrm{vec}\{\mathbf{A}\}$ denotes the vectorization operation that converts matrix $\mathbf{A}$ into a column vector by stacking its columns. $\|\cdot\|_{\mathrm{row},0}$ represents the number of non-zero rows of a matrix.
\section{System Model and Problem Formulation}\label{sec:sys}
\subsection{System Model}
As shown in Fig. \ref{fig:sysmodel}, we consider a MA-assisted MU-MIMO downlink system, where {one} BS equipped with $N$ MAs transmits the desired signal to $K$ users, each equipped with $M$ MAs in the presence of $R$ malicious jammers. We assume that the jammers adopt omnidirectional single antenna to impair the users’ signal receptions from all angles. In addition, a 2D Cartesian coordinate system is established to describe the positions of the transmit MAs at the BS and receive MAs at the users. Specifically, the positions of the $n$-th ($1 \leq n \leq N$) transmit antenna and the $m$-th ($1 \leq m \leq M$) receive antenna of user $k$ ($1 \leq k \leq K$) are denoted by $\mathbf{t}_{n} = [x^{\mathrm{t}}_{n}, y^{\mathrm{t}}_{n}]^{\mathrm{T}} \in \mathcal{C}_{\mathrm{t}}$ and $\mathbf{r}_{k,m} = [x^{\mathrm{r}}_{k,m}, y^{\mathrm{r}}_{k,m}]^{\mathrm{T}} \in \mathcal{C}_{\mathrm{r}, k}$, where $\mathcal{C}_{\mathrm{t}}$ and $\mathcal{C}_{\mathrm{r},k}$ are the given 2D movable regions of transmit and receive MAs, respectively. Without loss of generality, we set $\mathcal{C}_{\mathrm{t}}$ and $\mathcal{C}_{\mathrm{r},k}$ as square regions with size $W_\mathrm{t} \times W_\mathrm{t}$ and $W_\mathrm{r} \times W_\mathrm{r}$, respectively.
\par
Let $\mathbf{T} = [\mathbf{t}_{1}, \cdots, \mathbf{t}_{N}]^{\mathrm{T}}$ and $\mathbf{R}_{k} = [\mathbf{r}_{k, 1}, \cdots, \mathbf{t}_{k, M}]^{\mathrm{T}}$ denote the positions of the transmit MAs at the BS and the receive MAs at user $k$, respectively. Thus, the channels from the BS to user $k$ and from $r$-th jammer ($1\leq r\leq R$) to user $k$ can be denoted as $\mathbf{H}_{k}(\mathbf{T}, \mathbf{R}_{k}) \in \mathbb{C}^{M \times N}$ and $\mathbf{g}_{r,k}(\mathbf{R}_{k}) \in \mathbb{C}^{M \times 1}$, respectively. The desired signal transmitted to user $k$, denoted as $s_{k} \in \mathbb{C}$, is assumed to be independent random variables with $\mathbb{E}[s_{k}] =0$ and $\mathbb{E}[|s_{k}|^2] =1$. Prior to transmission, $s_{k}$ should be weighted by the transmit precoder $\mathbf{w}_{k} \in \mathbb{C}^{N \times 1}$, with the precoder being constrained by $\sum_{k=1}^{K}\Vert \mathbf{w}_{k} \Vert^{2} \leq P_{\mathrm{max}}$, where $P_{\mathrm{max}}$ denotes the maximum transmit power. Meanwhile, the $r$-th jammer transmits the jamming signal $x_{\mathrm{J}, r} = \sqrt{p_{\mathrm{J}, r}}s_{\mathrm{J}, r}$ to impair the signal reception, where $s_{\mathrm{J}, r}$ is normalized symbol and ${p_{\mathrm{J}, r}}$ is the corresponding jamming power. As such, the signal received by user $k$ is expressed as
\begin{equation}
	y_{k} = \mathbf{v}^{\mathrm{H}}_{k}\Big(\mathbf{H}_{k}(\mathbf{T}, \mathbf{R}_{k})\mathbf{x} + \sum\nolimits_{r=1}^{R}\mathbf{g}_{r,k}(\mathbf{R}_{k})x_{\mathrm{J}, r} + \mathbf{n}_{k}\Big)
\end{equation}
where $\mathbf{v}_{k} \in \mathbb{C}^{M \times 1}$ is the decoder adopted at the $k$-th user, $\mathbf{x} = \sum_{k=1}^{K}\mathbf{w}_{k}s_{k}$ is the transmit signal at the BS, and $\mathbf{n}_{k} \sim \mathcal{CN}(\mathbf{0}, \sigma^{2}_{k}\mathbf{I}_{M})$ represents the received noise.
\subsection{Channel Model}
In this subsection, we adopt field-response-based representation for channel modeling \cite{MA2}. Specifically, let $L^{\mathrm{t}}_{k}$ and $L^{\mathrm{r}}_{k}$ denote the numbers of legitimate transmit and receive paths, respectively. For the $p$-th path ($p=1,\cdots,L^{\mathrm{t}}_{k}$) of user $k$, the corresponding elevation and azimuth angles of departure (AoDs) are represented as $\theta^{\mathrm{t}}_{k, p}$ and $\phi^{\mathrm{t}}_{k, p}$, respectively. Then, the propagation difference  $\rho^{\mathrm{t}}_{k, p}(\mathbf{t}_{n})$ between the $n$-th transmit antenna and the reference point $\mathbf{o}_{\mathrm{t}} = [0, 0]^{\mathrm{T}}$ is given by $\rho^{\mathrm{t}}_{k, p}(\mathbf{t}_{n}) = x^{\mathrm{t}}_{n}\varphi^{\mathrm{t}}_{k, p} + y^{\mathrm{t}}_{n}\vartheta^{\mathrm{t}}_{k, p}$. Here, $\varphi^{\mathrm{t}}_{k, p} = \cos\theta^{\mathrm{t}}_{k, p}\sin\phi^{\mathrm{t}}_{k, p}$ and $\vartheta^{\mathrm{t}}_{k, p} = \sin\theta^{\mathrm{t}}_{k, p}$.
\par
As such, the phase difference of $k$-th user's $p$-th transmit path, relative to the reference point $\mathbf{o}_{\mathrm{t}}$, can be further expressed as $2\pi\rho^{\mathrm{t}}_{k, p}(\mathbf{t}_{n})/\lambda$, where $\lambda$ is wavelength. Therefore, the transmit field response vector for the $n$-th transmit MA is given by
\begin{equation}
	\mathbf{g}_{k}(\mathbf{t}_{n}) = \Big[e^{j\frac{2\pi}{\lambda}\rho^{\mathrm{t}}_{k, 1}(\mathbf{t}_{n})}, \cdots, e^{j\frac{2\pi}{\lambda}\rho^{\mathrm{t}}_{k, L^{\mathrm{t}}_{k}}(\mathbf{t}_{n})}\Big]^{\mathrm{T}}.
\end{equation}
\par
Similarly, the elevation and azimuth angles of arrival (AoAs) for the $q$-th ($q = 1, \cdots, L^{\mathrm{r}}_{k}$) receive path of user $k$ are denoted as $\theta^{\mathrm{r}}_{k, q}$ and $\phi^{\mathrm{r}}_{k, q}$, respectively. Then, the receive field response vector for the $m$-th receive MA of the $k$-th user is given by
\begin{equation}
	\mathbf{f}_{k}(\mathbf{r}_{k, m}) = \Big[e^{j\frac{2\pi}{\lambda}\rho^{\mathrm{r}}_{k, 1}(\mathbf{r}_{k, m})}, \cdots, e^{j\frac{2\pi}{\lambda}\rho^{\mathrm{r}}_{k, L^{\mathrm{r}}_{k}}(\mathbf{r}_{k, m})}\Big]^{\mathrm{T}}
\end{equation}
where $\rho^{\mathrm{r}}_{k, q}(\mathbf{r}_{k, m}) = x^{\mathrm{r}}_{k, m}\varphi^{\mathrm{r}}_{k, q} + y^{\mathrm{r}}_{k, m}\vartheta^{\mathrm{r}}_{k, q}$ with $\varphi^{\mathrm{r}}_{k, q} = \cos\theta^{\mathrm{r}}_{k, q}\sin\phi^{\mathrm{r}}_{k, q}$ and $\vartheta^{\mathrm{r}}_{k, q} = \sin\theta^{\mathrm{r}}_{k, q}$ quantifies the propagation delay for the $q$-th receive path relative to the reference point $\mathbf{o}_{\mathrm{r}, k} = [0, 0]^{\mathrm{T}}$.
\par
Next, the path response matrix (PRM) between the reference points of the transmit and $k$-th user's receive regions, i.e., $\mathbf{o}_{\mathrm{t}}$ and $\mathbf{o}_{\mathrm{r}, k}$, is denoted as $\mathbf{\Sigma}_{k} \in \mathbb{C}^{L^{\mathrm{r}}_{k} \times L^{\mathrm{t}}_{k}}$, with its entry characterizing the response between the corresponding transmit and receive paths. With above definitions, the channel matrix between the transmitter and $k$-th receiver can be accordingly expressed as
\begin{equation}\label{Eq:Channel}
	\mathbf{H}_{k} = \mathbf{F}^{\mathrm{H}}_{k}(\mathbf{R}_{k})\mathbf{\Sigma}_{k}\mathbf{G}_{k}(\mathbf{T})
\end{equation}
where $\mathbf{F}_{k}(\mathbf{R}_{k}) = \big[\mathbf{f}_{k}(\mathbf{r}_{k, 1}), \cdots, \mathbf{f}_{k}(\mathbf{r}_{k, M})\big] \in \mathbb{C}^{L^{\mathrm{r}}_{k} \times M}$ and $\mathbf{G}_{k}(\mathbf{T}) = \big[\mathbf{g}_{k}(\mathbf{t}_{1}), \cdots, \mathbf{g}_{k}(\mathbf{t}_{N})\big] \in \mathbb{C}^{L^{\mathrm{t}}_{k} \times N}$.
\par
Finally, we describe the channel matrix $\mathbf{g}_{r,k}(\mathbf{R}_{k})$ between the $r$-th jammer and $k$-th user. Similar to \eqref{Eq:Channel}, $\mathbf{g}_{r,k}(\mathbf{R}_{k})$ can be expressed as
\begin{equation}
	\mathbf{g}_{r,k}(\mathbf{R}_{k}) = \mathbf{F}^{\mathrm{H}}_{r, k}(\mathbf{R}_{k})\mathbf{\Sigma}_{r, k}\mathbf{1}_{L^{\mathrm{t}}_{r,k}}
\end{equation}
where $\mathbf{F}_{r, k}(\mathbf{R}_{k}) = \big[\mathbf{f}_{r,k}(\mathbf{r}_{k, 1}), \cdots, \mathbf{f}_{r,k}(\mathbf{r}_{k, M})\big]$ denotes field response matrix at the user $k$ with respect to the $r$-th jammer, and $\mathbf{f}_{r,k}(\mathbf{r}_{k, m}) = \big[e^{j\frac{2\pi}{\lambda}\rho^{\mathrm{r}}_{r,k,1}(\mathbf{r}_{k, m})}, \cdots, e^{j\frac{2\pi}{\lambda}\rho^{\mathrm{r}}_{r,k,L^{\mathrm{r}}_{r,k}}(\mathbf{r}_{k, m})}\big]^{\mathrm{T}}$. Therein, denote $L^{\mathrm{t}}_{r,k}$ and $L^{\mathrm{r}}_{r,k}$ as the number of transmit and receive paths from the $r$-th jammer to $k$-th user, respectively. In addition, $\theta^{\mathrm{r}}_{r,k,l}$ and $\phi^{\mathrm{r}}_{r,k,l}$ ($l = 1,\cdots,L^{\mathrm{r}}_{r,k}$) represent the corresponding elevation and azimuth AoAs for $l$-th receive path between jammer $r$ and user $k$, respectively. As such, we have $\rho^{\mathrm{r}}_{r,k,l}(\mathbf{r}_{k, m}) = x^{\mathrm{r}}_{k, m}\varphi^{\mathrm{r}}_{r, k, l} + y^{\mathrm{r}}_{k, m}\vartheta^{\mathrm{r}}_{r, k, l}$ with $\varphi^{\mathrm{r}}_{r, k, l} = \cos\theta^{\mathrm{r}}_{r, k, l}\sin\phi^{\mathrm{r}}_{r, k, l}$ and $\vartheta^{\mathrm{r}}_{r, k, l} = \sin\theta^{\mathrm{r}}_{r, k, l}$. Furthermore, $\mathbf{\Sigma}_{r, k} \in \mathbb{C}^{L^{\mathrm{r}}_{r,k} \times L^{\mathrm{t}}_{r,k}}$ represents the PRM between all transmit and receive paths.
\par
Since the jammers are not expected to cooperate with the legitimate nodes for channel estimation, the jammers' azimuth and elevation AoAs can be only estimated by detecting the jamming power, thus leading to imperfect jammers’ CSI. To account for this, we assume that $\mathbf{g}_{r,k}$ has a given azimuth/elevation AoA range \cite{10847939}, which is expressed as
\begin{equation}
	\Delta = \big\{\mathbf{g}_{r,k} \lvert \theta \in [\theta_{{L}}, \theta_{{U}}], \phi \in [\phi_{{L}}, \phi_{{U}}]\big\}
\end{equation}
where $\theta_{{L}}$ and $\theta_{{U}}$ denote the lower and upper bounds of azimuth angle, $\phi_{{L}}$ and $\phi_{{U}}$ denote the lower and upper bounds of elevation angle, respectively.\footnote{The same angular uncertainty set is adopted for the LoS and NLoS paths for two reasons. First, jamming in rich-scattering environments requires higher transmit power \cite{9367220}, making it reasonable to assume a LoS-dominated jamming link. Second, both paths are affected by the same jammer localization error \cite{10847939}, allowing $\Delta$ to jointly characterize their CSI inaccuracies. Moreover, random angular deviations in NLoS links can also be modeled through array-based phase rotations \cite{7982782}, further supporting this common uncertainty model.}
\subsection{Problem Formulation}
The received SINR of user $k$ is given by 
\begin{equation}
	\gamma_{k} = \frac{\vert \mathbf{v}^{\mathrm{H}}_{k}\mathbf{H}_{k}\mathbf{w}_{k} \vert^2}{\sum_{i \neq k}^{K}\vert \mathbf{v}^{\mathrm{H}}_{k}\mathbf{H}_{k}\mathbf{w}_{i} \vert^2 + \sum_{r=1}^{R}p_{\mathrm{J}, r}\vert \mathbf{v}^{\mathrm{H}}_{k}\mathbf{g}_{r, k} \vert^{2} + \widetilde{\sigma}^{2}_{k}}
\end{equation}
where $\widetilde{\sigma}^{2}_{k} = {\sigma}^{2}_{k}\Vert \mathbf{v}_{k}\Vert^{2}$. As such, the information rate of user $k$ can be expressed as $R_{k} = \log(1 + \gamma_{k})$. Next, the sum-rate maximization problem can be formulated as 
\begin{align}\label{Opt:Prop1}
	\mathop{\max}_{\mathbf{w}_{k}, \mathbf{v}_{k}, \mathbf{T}, \mathbf{R}_{k}} \ & \ \sum_{k=1}^{K} R_{k} \\
	\text{s.t.} \ \mathrm{C0:} & \mathbf{g}_{r,k} \in \Delta, \forall k,r \notag \\
	\mathrm{C1:} &  R_{k} \geq \Gamma_{k} \notag \\
	\mathrm{C2:} & \sum\nolimits_{k=1}^{K} \Vert \mathbf{w}_{k} \Vert^{2} \leq P_{\mathrm{max}} \notag \\
	\mathrm{C3:} & \mathbf{T} \in \mathcal{C}_{\mathrm{t}} \notag \\
	\mathrm{C4:} & \mathbf{R}_{k} \in \mathcal{C}_{\mathrm{r}, k} \notag \\
	\mathrm{C5:} & \Vert\mathbf{t}_{n} - \mathbf{t}_{n'}\Vert \geq D , 1 \leq n \neq n' \leq N \notag \\
	\mathrm{C6:} & \Vert\mathbf{r}_{k, m} - \mathbf{r}_{k,m'}\Vert \geq D, \forall k, 1 \leq m \neq m' \leq M. \notag
\end{align}
Here, $\mathrm{C1}$ is the minimum rate constraint with the $k$-th user’s target $\Gamma_{k}$, $\mathrm{C2}$ is the maximum total power constraint at BS, $\mathrm{C3}$--$\mathrm{C4}$ guarantee the MAs at the BS and users remain within the movable regions, and $\mathrm{C5}$--$\mathrm{C6}$ prevent coupling effects between any pair of MAs at the BS and users, respectively.
\section{Continuous Antenna Position Design:\\ An Alternating Optimization Framework}\label{sec:conti}
In this section, we develop an AO framework for continuous antenna position design in MA-aided anti-jamming communications. {The proposed continuous design framework exploits antenna-position flexibility for anti-jamming transmission and provides a continuous-position benchmark for the discrete-position design presented in Section~\ref{sec:discre}.}
\subsection{The Proposed AO Framework}
We adopt an alternating optimization framework to solve problem \eqref{Opt:Prop1} by decomposing it into four subproblems, corresponding to the decoder $\mathbf{v}_{k}$, the receive antenna positions $\mathbf{R}_{k}$, the precoder $\mathbf{w}_{k}$, and the transmit antenna positions $\mathbf{T}$. At each iteration, we sequentially update each variable block while keeping the others fixed, leveraging the MMSE criterion, Dinkelbach's method, SCA, and the MM framework as appropriate.
\subsection{Decoder and Continuous Receive Antenna Position Design under Angular Uncertainty}\label{Sec:ThreeA}
\subsubsection{{Method for Handling Jamming Angular Uncertainty}} There are several methods that have been developed to address angular uncertainty in MA systems. Specifically, the methods in \cite{11192669,11389850} derive uncertainty bounds induced by angular estimation errors and employ the S-procedure to reformulate the semi-infinite robust constraints as linear matrix inequalities. Alternatively, the methods in \cite{10857426,11603908} discretize angular uncertainty regions into finite samples and optimize the performance over these samples. These approaches, however, can incur substantial computational overhead. In this paper, we first uniformly discretize all the angles inside $\Delta$, i.e.,
\begin{align}
	& \theta^{(i)} = \theta_{L} + (i-1)\Delta\theta, i = 1,\cdots,Q_{1} \notag \\
	& \phi^{(j)} = \phi_{L} + (j-1)\Delta\phi, j = 1,\cdots,Q_{2}
\end{align}
where $Q_{1}$ and $Q_{2}$ are the number of samples of $\theta$ and $\phi$, $\Delta\theta = (\theta_{U} - \theta_{L}) / (Q_{1} - 1)$, and $\Delta\phi = (\phi_{U} - \phi_{L}) / (Q_{2} - 1)$. Next, we can construct the jamming covariance matrix $\mathbf{G}_{r,k}$ based on this sample element. Specifically, $\mathbf{G}_{r,k}$ is expressed as
\begin{equation}
	{\mathbf{G}}_{r, k} = \sum_{i=1}^{Q_{1}}\sum_{j=1}^{Q_{2}}\frac{1}{Q_1 Q_2}\widehat{\mathbf{g}}^{(i,j)}_{r, k}\widehat{\mathbf{g}}^{(i,j), \mathrm{H}}_{r, k}.
\end{equation}
Here, the uniform weights can effectively cover various cases within the uncertainty domain, which has been demonstrated to exhibit satisfactory robustness~\cite{10847939,10296481,11417830}.
\subsubsection{{MMSE-Based Robust Decoder Design}} 
First, we focus on investigating the design of $\mathbf{v}_{k}$. Since the linear MMSE detector is the optimal solution for $\mathbf{v}_{k}$, we can obtain
\begin{equation}\label{Eq:MMSE}
	\mathbf{v}_{k} = \bigg(\!\mathbf{H}_{k}\mathbf{W}\mathbf{W}^{\mathrm{H}}\mathbf{H}^{\mathrm{H}}_{k} + \sum_{r=1}^{R}p_{\mathrm{J}, r}{\mathbf{G}}_{r, k} + {\sigma}^{2}_{k}\mathbf{I}_{M}\!\bigg)^{-1}\mathbf{H}_{k}\mathbf{w}_{k}
\end{equation}
where $\mathbf{W} = [\mathbf{w}_{1}, \mathbf{w}_{2}, \cdots, \mathbf{w}_{K}] \in \mathbb{C}^{N \times K}$.

\subsubsection{{SCA-Based Robust Receive Antenna Position Design}}
Now, we turn to investigate the optimization of receive antenna position $\mathbf{R}_{k}$, whose corresponding subproblem can be formulated as
\begin{align}\label{Opt:Prop2}
	\mathop{\max}_{\mathbf{R}_{k}} \ &  \ \gamma_{k} \\
	\text{s.t.} \ & \mathrm{C4}, \mathrm{C6}. \notag
\end{align}
Then, we adopt the Dinkelbach’s method to transform \eqref{Opt:Prop2} into an equivalent form, i.e.,
\begin{align}\label{Opt:Prop3}
	\mathop{\min}_{\mathbf{R}_{k}} \ & f(\mathbf{R}_{k}, \kappa) \\
	\text{s.t.} \ & \mathrm{C4}, \mathrm{C6}\notag
\end{align}
where {$\kappa$ is the Dinkelbach auxiliary scalar}, and $f(\mathbf{R}_{k})$ is defined as 
\begin{align}\label{Eq:RecAntenna}
	f(\mathbf{R}_{k}) = & \kappa\sum\nolimits_{i \neq k}^{K}\vert \mathbf{v}^{\mathrm{H}}_{k}\mathbf{H}_{k}(\mathbf{R}_{k})\mathbf{w}_{i} \vert^2 - \vert \mathbf{v}^{\mathrm{H}}_{k}\mathbf{H}_{k}(\mathbf{R}_{k})\mathbf{w}_{k} \vert^2 \notag \\
	& + \kappa\sum\nolimits_{r=1}^{R}p_{\mathrm{J}, r}\mathbf{v}^{\mathrm{H}}_{k}{\mathbf{G}}_{r, k}(\mathbf{R}_{k})\mathbf{v}_{k}.
\end{align}
The subproblem w.r.t $\kappa$ admits $\kappa = \widetilde{\gamma}_{k}$, where $\widetilde{\gamma}_{k}$ is obtained by replacing $\mathbf{g}_{r,k}\mathbf{g}^{\mathrm{H}}_{r,k}$ with ${\mathbf{G}}_{r,k}$ in the expression of $\gamma_k$.
Since the problem w.r.t $\mathbf{R}_{k}$ is still non-convex and difficult to solve, we resort to SCA \cite{11142587}. Denote $\mathbf{r}_{k} = \mathrm{vec}\{\mathbf{R}_{k}\}$, we have the following proposition to transform the problem \eqref{Opt:Prop3} into a convex one.

\begin{proposition}\label{Prop:SCA}
	For arbitrary vectors $\mathbf{a} \in \mathbb{C}^{M}$ and $\mathbf{b} \in \mathbb{C}^{L^{\mathrm{r}}_{k}}$, concave and convex quadratic surrogate functions for $\vert \mathbf{b}^{\mathrm{H}}\mathbf{F}_{k}(\mathbf{R}_{k})\mathbf{a} \vert^{2}$ can be expressed as $\mathbf{r}_{k}^{\mathrm{T}} \mathbf{U}_{1} \mathbf{r}_{k} - \mathbf{r}_{k}^{\mathrm{T}}\bm{\nu}_{1} + d_{1}$ and $\mathbf{r}_{k}^{\mathrm{T}} \mathbf{U}_{2} \mathbf{r}_{k} - \mathbf{r}_{k}^{\mathrm{T}}\bm{\nu}_{2} + d_{2}$, respectively.
\end{proposition}

\begin{IEEEproof}
	See Appendix~\ref{App:A}.
\end{IEEEproof}

Building upon {Proposition~\ref{Prop:SCA}}, we can recast problem \eqref{Opt:Prop3} as the following problem
\begin{align}\label{Opt:Prop4}
	\mathop{\min}_{\mathbf{r}_{k}} \ & \mathbf{r}_{k}^{\mathrm{T}} \mathbf{U}_{k} \mathbf{r}_{k} - \mathbf{r}_{k}^{\mathrm{T}}\bm{\nu}_{k}  \\
	\text{s.t.} \ & \mathrm{C4}, \mathrm{\overline{C}6:} \frac{({\mathbf{r}}^{(i_{d})}_{k,m} - {\mathbf{r}}^{(i_{d})}_{k,m'})({\mathbf{r}}_{k,m} - {\mathbf{r}}_{k,m'})^{\mathrm{T}}}{\Vert {\mathbf{r}}^{(i_{d})}_{k,m} - {\mathbf{r}}^{(i_{d})}_{k,m'} \Vert} \geq D \notag
\end{align}
where ${\mathbf{r}}^{(i_{d})}_{k,m}$ is the solution obtained in the previous iteration, $\mathbf{U}_{k}$ and $\bm{\nu}_{k}$ are obtained by applying {Proposition~\ref{Prop:SCA}} term-by-term to \eqref{Eq:RecAntenna} and aggregating the resulting surrogate terms, and constraint $\mathrm{\overline{C}6}$ is obtained via first-order Taylor expansion.

Under the iterative framework, $\kappa$ and $\mathbf{R}_{k}$ can be optimized until converging to a stationary point \cite{10296481}, which is summarized in {Algorithm \ref{Alg:SupSCA}}. 
\begin{algorithm}[t]
	\footnotesize
	\caption{SCA for Antenna-Position Optimization}
	\label{Alg:SupSCA}
	\normalsize
	\begin{algorithmic}[1]
		\REQUIRE $\mathbf W$, $\mathbf v_k$, $\mathbf T$, and $\mathbf R_k^{(0)}$;
		\ENSURE $\mathbf R_k^{\star}$;
		\STATE Set $i_{d}\Leftarrow 0$ and $\mathbf r_k^{(0)}\Leftarrow\operatorname{vec}\{\mathbf R_k^{(0)}\}$;
		\REPEAT
		\STATE Update $\kappa$ as $\kappa=\widetilde{\gamma}_{k} \bigl(\mathbf R_{k}^{(i_{d})}\bigr)$;
		\STATE Calculate $\mathbf U_{k}^{(i_{d})}$ and $\bm{\nu}_{k}^{(i_{d})}$ according to Proposition~\ref{Prop:SCA};
		\STATE Update $\mathbf r_{k}^{(i_{d}+1)}$ by solving \eqref{Opt:Prop4};
		\STATE Recover
		$\mathbf R_k^{(i_{d}+1)}
		=\operatorname{unvec}\{\mathbf r_k^{(i_{d}+1)}\}$;
		\STATE Set $i_{d}\Leftarrow i_{d}+1$;
		\UNTIL{convergence or the maximum number of iterations is reached.}
		\STATE $\mathbf R_k^{\star}\Leftarrow\mathbf R_k^{(i_{d})}$;
	\end{algorithmic}
\end{algorithm}
\setlength{\textfloatsep}{5pt}
\vspace{-4mm}
\subsection{MM Algorithm for Precoder and Continuous Transmit Antenna Position Design}\label{Sec:ThreeB}
\subsubsection{{Problem Reformulation}} 
Here, we focus on the optimization of $\mathbf{W}, \mathbf{T}$ with given $\mathbf{v}_{k}, \mathbf{R}_{k}$. According to \cite{10054084}, any locally optimal solution must satisfy the transmit power constraint with equality.  This approach allows for solving the problem without the maximum transmit power constraint, followed by scaling the results to comply with the power constraint. In this context, the equivalent rate expression $\widetilde{R}_{k}$ for user $k$ incorporates the power constraint as follows:
\begin{align}
	\widetilde{R}_{k} = \log\bigg(1 + \frac{\vert \mathbf{h}^{\mathrm{H}}_{k}\mathbf{w}_{k} \vert^{2}}{\sum_{i \neq k}^{K}\vert \mathbf{h}^{\mathrm{H}}_{k}\mathbf{w}_{i} \vert^{2} + {\overline{\sigma}^{2}_{k}}\Vert \mathbf{W}  \Vert^{2}_{F}/P_{\mathrm{max}}}\bigg)
\end{align}
where $\mathbf{h}_{k} = \mathbf{H}^{\mathrm{H}}_{k}\mathbf{v}_{k}$, and $\overline{\sigma}^{2}_{k} = \sum_{r=1}^{R}p_{\mathrm{J}, r}\mathbf{v}^{\mathrm{H}}_{k}{\mathbf{G}}_{r, k}\mathbf{v}_{k} + {\sigma}^{2}_{k}\Vert \mathbf{v}_{k} \Vert^{2}$. Furthermore, the equivalent problem for $\mathbf{W}$ is shown as follows:
\begin{align}\label{Eq:Opt5}
	\mathop{\max}_{\mathbf{W}} \ &  \sum\nolimits_{k=1}^{K} \widetilde{R}_{k} \\
	\text{s.t.} \ &  \mathrm{\overline{C}1:} \ \widetilde{R}_{k} \geq \Gamma_{k}. \notag
\end{align}
Problem \eqref{Eq:Opt5} is still non-convex due to the expression for $\widetilde{R}_{k}$. Thus, the MM technique is adopted to convert $\widetilde{R}_{k}$ into solvable form. According to \cite{8769948}, a lower bound for $\widetilde{R}_{k}$ is given by
\begin{equation}\label{Eq:MM}
	\widetilde{R}_{k} \geq \log(\mathbf{u}^{\mathrm{H}}{\mathbf{D}}^{(i_{d}), -1}_{1,k}\mathbf{u}) - \tr\big({\mathbf{F}}^{(i_{d})}_{1,k}(\mathbf{D}_{1,k}-{\mathbf{D}}^{(i_{d})}_{1,k})\big)
\end{equation}
where 
\begin{align}
	& \mathbf{D}_{1, k} = \left[\begin{matrix}
		1 & \mathbf{w}^{\mathrm{H}}_{k}\mathbf{h}_{k} \\
		\mathbf{h}^{\mathrm{H}}_{k}\mathbf{w}_{k} & \mathbf{h}^{\mathrm{H}}_{k}\mathbf{W}\mathbf{W}^{\mathrm{H}}\mathbf{h}_{k} + {\overline{\sigma}^{2}_{k}}\Vert \mathbf{W}  \Vert^{2}_{F}/P_{\mathrm{max}}
	\end{matrix}\right] \notag \\
	& {\mathbf{F}}^{(i_{d})}_{1,k} =  {\mathbf{D}}^{(i_{d}),-1}_{1,k}\mathbf{u}\big(\mathbf{u}^{\mathrm{H}}{\mathbf{D}}^{(i_{d}), -1}_{1,k}\mathbf{u}\big)^{-1}\mathbf{u}^{\mathrm{H}}{\mathbf{D}}^{(i_{d}),-1}_{1,k}, \mathbf{u} = [1, 0]^{\mathrm{T}}  \notag
\end{align}
and $(\cdot)^{(i_{d})}$ denotes the solution obtained in the last iteration. Moreover, denoting ${\mathbf{w}}^{(i_{d})}_{k}$ and ${\mathbf{h}}^{(i_{d})}_{k}$ as the solution obtained in the previous iteration, ${\mathbf{F}}^{(i_{d})}_{1,k}$ can be further decomposed as
\begin{equation}\label{Eq:F}
	{\mathbf{F}}^{(i_{d})}_{1,k} = \left[\begin{matrix}
		\big[{\mathbf{F}}^{(i_{d})}_{1,k}\big]_{1,1} & \big[{\mathbf{F}}^{(i_{d})}_{1,k}\big]_{1,2} \\
		\big[{\mathbf{F}}^{(i_{d})}_{1,k}\big]_{2,1} & \big[{\mathbf{F}}^{(i_{d})}_{1,k}\big]_{2,2}
	\end{matrix}\right]
\end{equation}
where 
\begin{align}
	& \big[{\mathbf{F}}^{(i_{d})}_{1,k}\big]_{1,1} = 1 + \varpi\overline{\mathbf{h}}^{\mathrm{H}}_{k}{\mathbf{w}}^{(i_{d})}_{k}{\mathbf{w}}^{(i_{d}), \mathrm{H}}_{k}{\mathbf{h}}^{(i_{d}), \mathrm{H}}_{k} \notag \\
	& \big[{\mathbf{F}}^{(i_{d})}_{1,k}\big]_{1,2} = - \varpi{\mathbf{w}}^{(i_{d}), \mathrm{H}}_{k}{\mathbf{h}}^{(i_{d})}_{k}, \big[{\mathbf{F}}^{(i_{d})}_{1,k}\big]_{2,1} = - \varpi{\mathbf{h}}^{{(i_{d})},\mathrm{H}}_{k}{\mathbf{w}}^{(i_{d})}_{k} \notag \\
	& \big[{\mathbf{F}}^{(i_{d})}_{1,k}\big]_{2,2} = \big(\big[{\mathbf{F}}^{(i_{d})}_{1,k}\big]_{1,1}\big)^{-1}\big[{\mathbf{F}}^{(i_{d})}_{1,k}\big]_{2,1}\big[{\mathbf{F}}^{(i_{d})}_{1,k}\big]_{1,2} \notag \\
	& \varpi = {P_{\mathrm{max}}}/\big({{\overline{\sigma}^{2}_{k}}\Vert {\mathbf{W}^{(i_{d})}}  \Vert^{2}_{F}}\big).
\end{align}
Thus, $\tr\big({\mathbf{F}}^{(i_{d})}_{1,k}\mathbf{D}_{1,k}\big)$ can be transformed into following equivalent form, namely,
\begin{align}
	& \tr\big({\mathbf{F}}^{(i_{d})}_{1,k}\mathbf{D}_{1,k}\big) =  \big[{\mathbf{F}}^{(i_{d})}_{1,k}\big]_{1,1}  + 2\re(\big[{\mathbf{F}}^{(i_{d})}_{1,k}\big]_{1,2}\mathbf{h}^{\mathrm{H}}_{k}\mathbf{w}_{k})  \notag \\
	&  + \tr(\big[{\mathbf{F}}^{(i_{d})}_{1,k}\big]_{2,2}(\mathbf{h}^{\mathrm{H}}_{k}\mathbf{w}_{k}\mathbf{w}^{\mathrm{H}}_{k}\mathbf{h}_{k} + {\overline{\sigma}^{2}_{k}}\Vert \mathbf{W}  \Vert^{2}_{F}/P_{\mathrm{max}})).
\end{align}
Via the above procedure and dropping the constant terms, the problem can be reformulated as
\begin{align}\label{Eq:Opt6}
	\mathop{\min}_{\mathbf{W}, \mathbf{T}} \ &  \sum\nolimits_{k=1}^{K} 2\re\big(\big[{\mathbf{F}}^{(i_{d})}_{1,k}\big]_{1,2}\mathbf{h}^{\mathrm{H}}_{k}(\mathbf{T})\mathbf{w}_{k}\big)  \\
	&  + \tr\big(\big[{\mathbf{F}}^{(i_{d})}_{1,k}\big]_{2,2}\mathbf{h}^{\mathrm{H}}_{k}(\mathbf{T})\mathbf{W}\mathbf{W}^{\mathrm{H}}\mathbf{h}_{k}(\mathbf{T})\big) \notag \\
	&  + \tr\big(\big[{\mathbf{F}}^{(i_{d})}_{1,k}\big]_{2,2}{\overline{\sigma}^{2}_{k}}\Vert \mathbf{W}  \Vert^{2}_{F}/P_{\mathrm{max}}\big) \notag \\
	\text{s.t.} \ & \mathrm{\overline{\overline{C}}1:} 2\re\big(\big[{\mathbf{F}}^{(i_{d})}_{1,k}\big]_{1,2}\mathbf{h}^{\mathrm{H}}_{k}(\mathbf{T})\mathbf{w}_{k}\big) \notag \\
	&  + \tr\big(\big[{\mathbf{F}}^{(i_{d})}_{1,k}\big]_{2,2}\mathbf{h}^{\mathrm{H}}_{k}(\mathbf{T})\mathbf{W}\mathbf{W}^{\mathrm{H}}\mathbf{h}_{k}(\mathbf{T})\big) \notag \\
	&  + \tr\big(\big[{\mathbf{F}}^{(i_{d})}_{1,k}\big]_{2,2}{\overline{\sigma}^{2}_{k}}\Vert \mathbf{W}  \Vert^{2}_{F}/P_{\mathrm{max}}\big) \leq \widetilde{\Gamma}_{k}, \forall k \notag 
\end{align}
where $\widetilde{\Gamma}_{k} = \log(\mathbf{u}^{\mathrm{H}}{\mathbf{D}}_{1,k}^{{(i_{d})},-1}\mathbf{u}) + \tr\big({\mathbf{F}}^{(i_{d})}_{1,k}{\mathbf{D}}^{(i_{d})}_{1,k}\big) - \Gamma_{k} - \big[{\mathbf{F}}^{(i_{d})}_{1,k}\big]_{1,1}$.
\subsubsection{{QCQP Optimization for $\mathbf{W}$}}
Note problem \eqref{Eq:Opt6} is a quadratic constraint quadratic programming (QCQP), which can be solved using the CVX optimization toolbox efficiently. 
\subsubsection{{SCA-Based Transmit Antenna Position Design}} {Proposition~\ref{Prop:SCA}} suggests that the non-convex quadratic term in problem~\eqref{Eq:Opt6} w.r.t $\mathbf{T}$ can be approximated by a convex quadratic surrogate via the SCA framework. Then, we turn to solve the nonconvexity of the term $2\operatorname{Re}\big( \big[{\mathbf{F}}^{(i_{d})}_{1,k}\big]_{1,2} \, \mathbf{h}^{\mathrm{H}}_{k}(\mathbf{T}) \mathbf{w}_{k}\big)$ with respect to $\mathbf{T}$. By denoting $\mathbf{t} = \mathrm{vec}\{\mathbf{T}\}$, we have the following proposition.
\begin{proposition}\label{Prop:SCA2}
	For arbitrary vectors $\mathbf{a} \in \mathbb{C}^{N}$ and $\mathbf{b} \in \mathbb{C}^{L^{\mathrm{t}}_{k}}$, a convex quadratic surrogate function for $\re(\mathbf{b}^{\mathrm{T}}\mathbf{G}_{k}(\mathbf{T})\mathbf{a})$ is given by $\mathbf{t}^{\mathrm{H}}\mathbf{\Omega}\mathbf{t} - \mathbf{t}^{\mathrm{T}}\bm{\mu} + c$, where $\mathbf{\Omega}$ is positive semidefinite.
\end{proposition}
\begin{IEEEproof}
	See Appendix~\ref{App:B}.
\end{IEEEproof}
\par
Armed with {Proposition~\ref{Prop:SCA2}}, the subproblem to $\mathbf{T}$ can be recast as
\begin{align}\label{Opt:Prop7}
	\mathop{\min}_{\mathbf{t}} \ & \sum_{k=1}^{K}\mathbf{t}^{\mathrm{T}} \mathbf{\Omega}_{k} \mathbf{t} - \mathbf{t}^{\mathrm{T}}\bm{\mu}_{k}  \\
	\text{s.t.} \ & \mathrm{{\underline{C}}1:} \mathbf{t}^{\mathrm{T}} \mathbf{\Omega}_{k} \mathbf{t} - \mathbf{t}^{\mathrm{T}}\bm{\mu}_{k} + c_{k} \leq \widehat{\Gamma}_{k} \notag \\
	\ & \mathrm{C3}, \mathrm{\overline{C}5:} \frac{({\mathbf{t}}^{(i_{d})}_{n} - {\mathbf{t}}^{(i_{d})}_{n'})({\mathbf{t}}_{n} - {\mathbf{t}}_{n'})^{\mathrm{T}}}{\Vert {\mathbf{t}}^{(i_{d})}_{n} - {\mathbf{t}}^{(i_{d})}_{n'} \Vert} \geq D \notag
\end{align}
where $\widehat{\Gamma}_{k} = \widetilde{\Gamma}_{k} - \big[{\mathbf{F}}^{(i_{d})}_{1,k}\big]_{2,2}{\overline{\sigma}^{2}_{k}}\Vert \mathbf{W}  \Vert^{2}_{F}/P_{\mathrm{max}}$, $\mathbf{\Omega}_{k}$, $\bm{\mu}_{k}$ and $c_{k}$ are obtained by applying {Proposition~\ref{Prop:SCA}} and {\ref{Prop:SCA2}} term-by-term to $\mathrm{\overline{\overline{C}}1}$, then aggregating the resulting surrogate terms, and ${\mathbf{t}}^{(i_{d})}_{n}$ is the solution obtained in the previous iteration. Consequently, problem~\eqref{Opt:Prop7} can be efficiently solved via the SCA framework outlined in {Algorithm~\ref{Alg:SupSCA}}, which is guaranteed to converge to a stationary point.
\begin{algorithm}[!t]
	\normalsize
	\caption{AO Algorithm for Problem \eqref{Opt:Prop1}}
	\label{Alg:AO}
	\begin{algorithmic}[1]
		\REQUIRE CSI (AoA, AoD, and PRM);
		\ENSURE $\mathbf{w}_{k}$, $\mathbf{v}_{k}$, $\mathbf{T}$, and $\mathbf{R}_{k}$;
		\STATE $i_{d} \Leftarrow 0$;
		\REPEAT
		\STATE Update $\mathbf{v}_{k}$ according to \eqref{Eq:MMSE};
		\REPEAT
		\STATE Update $\mathbf{R}_{k}$ using the SCA inner loop;
		\UNTIL{the objective value of problem \eqref{Opt:Prop3} converges.}
		\STATE Calculate ${\mathbf{D}}^{(i_{d})}_{1,k}$ and ${\mathbf{F}}^{(i_{d})}_{1,k}$;
		\STATE Update $\mathbf{W}$ by solving problem \eqref{Eq:Opt6};
		\REPEAT
		\STATE Update $\mathbf{T}$ using the SCA inner loop;
		\UNTIL{the objective value of problem \eqref{Eq:Opt6} converges.}
		\STATE Set $\mathbf{W}^{\star} \Leftarrow \sqrt{P_{\mathrm{max}}}\mathbf{W}/\Vert \mathbf{W} \Vert_{F}$;
		\STATE $i_{d} \Leftarrow i_{d} + 1$;
		\UNTIL{the objective values of \eqref{Opt:Prop1} converges or the iteration  number reaches the maximum.}
	\end{algorithmic}
\end{algorithm}
\setlength{\textfloatsep}{5pt}
\subsection{Initialization}
It is worth noting that the proposed SCA algorithm is sensitive to the choice of the initial antenna position. Therefore, an appropriate initialization of the MA positions is crucial to the convergence performance and solution quality of the AO algorithm. To this end, a circle-packing scheme~\cite{10243545} is adopted to initialize the MA positions. This scheme offers two main advantages. First, increasing the pairwise distances between the MA can effectively mitigate the mutual coupling effect. Second, distributing the MA sufficiently and uniformly over the feasible region enables more effective exploitation of the available spatial DoF.
\subsection{Complexity Analysis}
We summarize the overall AO algorithm in {Algorithm~\ref{Alg:AO}}. The MMSE receiver update costs $C_{\mathbf v}=\mathcal{O}(KM^{3})$. With the path number $L$ and $Q=Q_{1}Q_{2}$ angular samples, the receive- and transmit-position SCA updates cost $C_{\mathbf R}=\mathcal{O}(I_{1}((K+QR)M^{2}L^{2}+M^{3}))$ and $C_{\mathbf T}=\mathcal{O}(I_{2}(KN^{2}L^{2}+N^{3}))$, respectively, while the precoder QCQP costs $C_{\mathbf W}=\mathcal{O}(K^{3}N^{3})$. Hence $\mathcal{O}_{\rm AO}=\mathcal{O}(\max\{C_{\mathbf v},C_{\mathbf R},C_{\mathbf T},C_{\mathbf W}\})$.
\section{Discrete Antenna Position Design: \\A Sparse Recovery Framework}\label{sec:discre}
In this section, we exploit the sparse structure induced by a finite candidate grid. By embedding position selection into row-sparse beamformers, the discrete design admits coupled updates of beamforming and antenna positions with lower implementation complexity. Specifically, we propose a heuristic block coordinate descent (BCD) method augmented with RLS-SOMP for the joint optimization of beamforming and discrete antenna position selection.
\subsection{The Proposed Sparse Recovery Framework}
We first discretize the transmitter and user~$k$'s moving regions $\mathcal{C}_{\mathrm{t}}$ and $\mathcal{C}_{\mathrm{r},k}$ into $G_{\mathrm{t}}$ and $G_{\mathrm{r},k}$ candidate antenna positions with spacing $D$, thereby obtaining a grid-based mixed-integer approximation of problem~\eqref{Opt:Prop1}
\begin{align}\label{Opt:PropDis}
	\mathop{\max}_{\mathbf{w}_{k}, \mathbf{v}_{k}, \mathbf{B}, \mathbf{C}_{k}} \ & \ \sum\nolimits_{k=1}^{K} \widehat{R}_{k} \\
	\text{s.t.} \ \mathrm{\widehat{C}0:}& \hat{\mathbf{g}}_{r,k} \in \Delta, \forall k,r \notag \\
	\mathrm{\widehat{C}1:} & \widehat{R}_{k} \geq \Gamma_{k} \notag  \\
	\mathrm{\widehat{C}3:} & \mathbf{B} \in \{0,1\}^{G_{\mathrm{t}} \times N}, \mathbf{1}^{\mathrm{T}} \mathbf{B} = \mathbf{1}^{\mathrm{T}}, \mathbf{B} \mathbf 1_N\preceq\mathbf 1_{G_{\mathrm{t}}} \notag \\
	\mathrm{\widehat{C}4:} & \mathbf{C}_{k} \in \{0,1\}^{G_{\mathrm{r}, k} \times M}, \mathbf{1}^{\mathrm{T}} \mathbf{C}_{k} = \mathbf{1}^{\mathrm{T}}, \mathbf{C}_k\mathbf 1_M\preceq\mathbf 1_{G_{\mathrm{r},k}} \notag 
\end{align}
where $\mathbf{B}$ and $\mathbf{C}_{k}$ denote the binary selection matrices, and
\begin{align}
	& \widehat{R}_{k} = \log(1 + \widehat{\gamma}_{k}),\ {\widehat{\sigma}^{2}_{1, k}} = \sum\nolimits_{r=1}^{R}p_{\mathrm{J}, r}\vert \mathbf{v}^{\mathrm{H}}_{k}\mathbf{C}^{\mathrm{H}}_{k}\widehat{\mathbf{g}}_{r, k} \vert^{2} + \widetilde{\sigma}^{2}_{k} \notag \\
	& \widehat{\gamma}_{k} = \frac{\vert \mathbf{v}^{\mathrm{H}}_{k}\mathbf{C}^{\mathrm{H}}_{k}\widehat{\mathbf{H}}_{k}\mathbf{B}\mathbf{w}_{k} \vert^{2}}{\sum_{i \neq k}^{K}\vert \mathbf{v}^{\mathrm{H}}_{k}\mathbf{C}^{\mathrm{H}}_{k}\widehat{\mathbf{H}}_{k}\mathbf{B}\mathbf{w}_{i} \vert^{2} + {\widehat{\sigma}^{2}_{1, k}}\Vert \mathbf{W}  \Vert^{2}_{F}/P_{\mathrm{max}}}
\end{align}
with $\widehat{\mathbf{H}}_{k} \in \mathbb{C}^{G_{\mathrm{r}, k} \times G_{\mathrm{t}}}$ and $\widehat{\mathbf{g}}_{r,k}  \in \mathbb{C}^{G_{\mathrm{r}, k} \times 1}$ constructed from the discretized antenna positions.
\par
Then, we adopt a heuristic BCD framework to solve problem \eqref{Opt:PropDis} by decomposing it into three blocks, corresponding to the transmitter variables (i.e., $\mathbf{W}$ and $\mathbf{B}$), the receiver variables (i.e., $\mathbf{v}_{k}$ and $\mathbf{C}_{k}$), and the transmit power. At each iteration, we sequentially update one block while keeping the others fixed, where the variables within each block are optimized jointly. Specifically, the transmitter and receiver blocks are solved using the RLS-SOMP method, whereas the transmit power is updated via a dedicated power optimization subproblem.
\subsection{MMSE-SOMP for Robust Decoder and Discrete Receive Antenna Position Design}
This part jointly optimizes $\mathbf{v}_{k}$ and $\mathbf{C}_{k}$. First, the MSE of user $k$ is expressed as
\begin{align}\label{Eq:MSE}
	e_k = &
	{| 1 - \mathbf{v}_k^{\mathrm{H}} \mathbf{C}_k^{\mathrm{H}} \widehat{\mathbf{H}}_k \mathbf{B} \mathbf{w}_k |^2}
	+ {\omega \,
		\mathbf{v}_k^{\mathrm{H}} \mathbf{C}_k^{\mathrm{H}} 
		\mathbf{O}_{\mathrm{J}, k}
		\mathbf{C}_k \mathbf{v}_k}\notag\\
	&+ {\sum\nolimits_{\substack{i\neq k}}^{K}
		| \mathbf{v}_k^{\mathrm{H}} \mathbf{C}_k^{\mathrm{H}} \widehat{\mathbf{H}}_k \mathbf{B} \mathbf{w}_i |^2}+ {\omega\sigma_k^2\|\mathbf{v}_k \|^2}
\end{align}
where $\omega = {\tr\big(\mathbf{W}\mathbf{W}^{\mathrm{H}}\big)}/{P_{\max}}$ is scaling factor, and $\mathbf{O}_{\mathrm{J}, k} = \sum_{r=1}^{R} \sum_{i=1}^{Q} {p_{\mathrm{J},r}} \,
\overline{\mathbf{g}}_{r,k}^{(i)} \overline{\mathbf{g}}_{r,k}^{(i),\mathrm{H}} / {Q}$
with $\overline{\mathbf{g}}^{(i)}_{r, k}$ being the sampling channel. With \eqref{Eq:MSE}, minimizing the MSE leads to the following problem:
\begin{align}\label{Opt:PropDisMMSE}
	\{\mathbf{v}^{\star}_{k}, \mathbf{C}^{\star}_{k}\} = \mathop{\arg\min}_{\mathbf{v}_{k}, \mathbf{C}_{k}} \ e_{k} \ \  \text{s.t.} \ \mathrm{\widehat{C}4}.
\end{align}
Nevertheless, problem \eqref{Opt:PropDisMMSE} is also intractable due to the high-dimensional binary selection matrix. Thus, we have the following proposition for reformulating the objective.
\begin{proposition}\label{Prop:MMSE-SOMP}
	The MMSE detector problem is equivalent to the following $\ell_{2}$ norm regularized least squares problem as
	\begin{align}\label{Prob:LP_OPT_Re}
		\ \mathop{\min}_{\mathbf{v}_{k}, \mathbf{C}_{k}} &  \ \Vert {\mathbf{e}}_{k} - \mathbf{E}_{k}\mathbf{C}_{k}\mathbf{v}_{k} \Vert^{2} + \zeta_{1}\Vert \mathbf{v}_{k} \Vert^{2} \\
		\qquad \text{s.t.} & \ \mathrm{\widehat{C}4} \notag 
	\end{align}
	where ${\mathbf{e}}_{k} \in \mathbb{R}^{K + RQ}$ is the $k$-th standard basis vector, $\zeta_{1} = \omega{\sigma^{2}_{k}}$, and $\mathbf{E}_{k}$ is given by
	\begin{equation}
		\mathbf E_k=\big[\widehat{\mathbf H}_k\mathbf B\mathbf W,
		\alpha_1\mathbf J_{1,k},\ldots,\alpha_R\mathbf J_{R,k}\big]^{\mathrm H}.
	\end{equation}
	Here, $\alpha_r=\sqrt{p_{\mathrm J,r}\omega/Q}$ and $\mathbf J_{r,k}=[\overline{\mathbf g}_{r,k}^{(1)},\ldots,\overline{\mathbf g}_{r,k}^{(Q)}]$.
\end{proposition}
\begin{IEEEproof}
	See Appendix~\ref{App:C}.
\end{IEEEproof}
\par
In light of {Proposition~\ref{Prop:MMSE-SOMP}}, by defining the sparse detector $\widehat{\mathbf{v}}_{k}\triangleq\mathbf{C}_{k}\mathbf{v}_{k}$, problem~\eqref{Prob:LP_OPT_Re} can be recast as a regularized sparse recovery problem according to \cite{F-WMMSE}, which is stated as follows
\begin{align}\label{Prob:MMSE-SOMP}
	\ \mathop{\min}_{\widehat{\mathbf{v}}_{k}} &  \ \Vert {\mathbf{e}}_{k} - \mathbf{E}_{k}\widehat{\mathbf{v}}_{k} \Vert^{2} + \zeta_{1}\Vert \widehat{\mathbf{v}}_{k} \Vert^{2} \\
	\qquad \text{s.t.} & \ \mathrm{\underline{C}4:}\Vert \widehat{\mathbf{v}}_{k} \Vert_{\mathrm{row}, 0} = M. \notag 
\end{align}
{The equivalence between \eqref{Prob:LP_OPT_Re} and \eqref{Prob:MMSE-SOMP} follows from \(\mathbf C_k=[\mathbf e_{\lambda_1},\ldots,\mathbf e_{\lambda_M}]\), which consists of \(M\) distinct standard basis vectors. Specifically, the structure of $\mathbf C_k$ yields $\|\widehat{\mathbf v}_k\|_{\mathrm{row},0}= M$, while \(\mathbf C_k^{\mathrm H}\mathbf C_k=\mathbf I_M\) ensures \(\|\widehat{\mathbf v}_k\|^2=\|\mathbf v_k\|^2\). Thus, optimizing $\widehat{\mathbf v}_k$ jointly determines the antenna positions through its support and the corresponding beamforming weights through its nonzero entries.}
\par
Notably, problem \eqref{Prob:MMSE-SOMP} is a regularized sparse recovery problem with sparsity level $M$, which can be efficiently solved via the RLS-SOMP outlined in {Algorithm~\ref{Alg:SupRLS}}. Then, we can obtain $\mathbf{C}_{k}$ according to set $\boldsymbol{\Lambda}_{\mathrm{r}, k}$, where $\boldsymbol{\Lambda}_{\mathrm{r}, k}$ collects the indices of the nonzero entries of \(\widehat{\mathbf v}_k\).
\begin{algorithm}[t]
	\normalsize
	\caption{RLS-SOMP for Regularized Row-Sparse Recovery Problem}\label{Alg:SupRLS}
	\begin{algorithmic}[1]
		\REQUIRE $\mathbf Y$, $\mathbf D$, $S$, and $\zeta$;
		\ENSURE $\mathbf X^{\star}$ and $\boldsymbol\Lambda$;
		\BlankLine
		
		\tcp*[h]{Solve: $\mathop{\arg\min}\limits_{\|\mathbf{X}\|_{\mathrm{row},0}=N}\|\mathbf{Y}-\mathbf{D}\mathbf{X}\|_{F}^{2}+\zeta\|\mathbf{X}\|_{F}^{2}$}
		\STATE $\mathbf R \Leftarrow \mathbf Y$, $\mathbf\Lambda \Leftarrow \emptyset$, $\Gamma \Leftarrow \{1,\ldots,G\}$;
		\FOR{$s=1$ to $S$}
		\STATE $g^{\star} \Leftarrow \arg\max_{g\in\Gamma}\|[\mathbf D]_{:,g}^{\mathrm H}\mathbf R\|_{2}^{2}$;
		\STATE $\boldsymbol\Lambda \Leftarrow \mathbf\Lambda\cup\{g^{\star}\}$, $\Gamma \Leftarrow \Gamma\setminus\{g^{\star}\}$;
		\STATE $\mathbf D_{\mathbf\Lambda} \Leftarrow [\mathbf D]_{:,\boldsymbol\Lambda}$;
		\STATE $\mathbf A_{\mathbf\Lambda} \Leftarrow \mathbf D_{\boldsymbol\Lambda}^{\mathrm H}\mathbf D_{\boldsymbol\Lambda}+\zeta\mathbf I$;
		\STATE $\mathbf X_{\mathbf\Lambda} \Leftarrow \mathbf A_{\boldsymbol\Lambda}^{-1}\mathbf D_{\boldsymbol\Lambda}^{\mathrm H}\mathbf Y$;
		\STATE $\mathbf R \Leftarrow \mathbf Y-\mathbf D_{\boldsymbol\Lambda}\mathbf X_{\boldsymbol\Lambda}$;
		\ENDFOR
		\STATE Set the rows outside $\boldsymbol\Lambda$ to zero and return $\mathbf X^{\star}$;
	\end{algorithmic}
\end{algorithm}
\setlength{\textfloatsep}{5pt}
\subsection{MM-SOMP for Precoder and Discrete Transmit Antenna Position Design}
In this part, the optimization of the precoder $\mathbf{W}$ and transmit antenna position selection matrix $\mathbf{B}$ are investigated. According to {Section~\ref{Sec:ThreeB}}, the achievable information rate for the $k$-th user can be rewritten as
\begin{equation}
	\widehat{R}_{k} = \log\bigg(1+\frac{\vert \widehat{\mathbf{h}}^{\mathrm{H}}_{k}\mathbf{B}\mathbf{w}_{k} \vert^{2}}{\sum_{i \neq k}^{K}\vert \widehat{\mathbf{h}}^{\mathrm{H}}_{k}\mathbf{B}\mathbf{w}_{i} \vert^{2} + \omega{\widehat{\sigma}^{2}_{2, k}}}\bigg)
\end{equation}
where $\widehat{\mathbf{h}}_{k} = \widehat{\mathbf{H}}^{\mathrm{H}}_{k}\mathbf{C}_{k}\mathbf{v}_{k}$, and ${\widehat{\sigma}^{2}_{2, k}} = \mathbf{v}_k^{\mathrm{H}} \mathbf{C}_k^{\mathrm{H}} 
\mathbf{O}_{\mathrm{J}, k}
\mathbf{C}_k \mathbf{v}_k + \widetilde{\sigma}^{2}_{k}$. Then, based on \eqref{Eq:MM}, we can obtain a tractable lower bound for $\widehat{R}_{k}$, which is given by
\begin{equation}
	\widehat{R}_{k} \geq \log(\mathbf{u}^{\mathrm{H}}{\mathbf{D}}^{(i_{d}), -1}_{2,k}\mathbf{u}) - \tr\big({\mathbf{F}}^{(i_{d})}_{2,k}(\mathbf{D}_{2,k}-{\mathbf{D}}^{(i_{d})}_{2,k})\big)
\end{equation}
where 
\begin{align}
	& \mathbf{D}_{2, k} = \left[\begin{matrix}
		1 & \mathbf{w}^{\mathrm{H}}_{k}\mathbf{B}^{\mathrm{H}}\widehat{\mathbf{h}}_{k} \\
		\widehat{\mathbf{h}}^{\mathrm{H}}_{k}\mathbf{B}\mathbf{w}_{k} & 	\widehat{\mathbf{h}}^{\mathrm{H}}_{k}\mathbf{B}\mathbf{W} \mathbf{W}^{\mathrm{H}}\mathbf{B}^{\mathrm{H}}\widehat{\mathbf{h}}_{k} + \omega{\widehat{\sigma}^{2}_{2, k}}
	\end{matrix}\right] \notag \\
	& {\mathbf{F}}^{(i_{d})}_{2,k} =  {\mathbf{D}}_{2,k}^{{(i_{d})},-1}\mathbf{u}\big(\mathbf{u}^{\mathrm{H}}{\mathbf{D}}_{2,k}^{{(i_{d})},-1}\mathbf{u}\big)^{-1}\mathbf{u}^{\mathrm{H}}{\mathbf{D}}_{2,k}^{{(i_{d})},-1}  \notag
\end{align}
and ${\mathbf{D}}^{(i_{d})}_{2,k}$ denotes the solution obtained in the last iteration. Similar to \eqref{Eq:F}, ${\mathbf{F}}^{(i_{d})}_{2,k}$ can be further decomposed as
\begin{equation}
	{\mathbf{F}}^{(i_{d})}_{2,k} = \left[\begin{matrix}
		\big[{\mathbf{F}}^{(i_{d})}_{2,k}\big]_{1,1} & \big[{\mathbf{F}}^{(i_{d})}_{2,k}\big]_{1,2} \\
		\big[{\mathbf{F}}^{(i_{d})}_{2,k}\big]_{2,1} & \big[{\mathbf{F}}^{(i_{d})}_{2,k}\big]_{2,2}
	\end{matrix}\right]
\end{equation}
where 
\begin{align}
	& \big[{\mathbf{F}}^{(i_{d})}_{2,k}\big]_{1,1} = 1 + \widehat{\mathbf{h}}^{\mathrm{H}}_{k}{\mathbf{B}}^{(i_{d})}{\mathbf{w}}^{(i_{d})}_{k}{\mathbf{w}}^{{(i_{d})},\mathrm{H}}_{k}{\mathbf{B}}^{{(i_{d})},\mathrm{H}}\widehat{\mathbf{h}}_{k} / \omega{\widehat{\sigma}^{2}_{2, k}} \notag \\
	& \big[{\mathbf{F}}^{(i_{d})}_{2,k}\big]_{1,2} = - {\mathbf{w}}^{{(i_{d})},\mathrm{H}}_{k}{\mathbf{B}}^{{(i_{d})},\mathrm{H}}\widehat{\mathbf{h}}_{k}/\omega{\widehat{\sigma}^{2}_{2, k}}\notag \\
	& \big[{\mathbf{F}}^{(i_{d})}_{2,k}\big]_{2,1} = - \widehat{\mathbf{h}}^{\mathrm{H}}_{k}{\mathbf{B}}^{(i_{d})}{\mathbf{w}}^{(i_{d})}_{k}/\omega{\widehat{\sigma}^{2}_{2, k}} \notag \\
	& \big[{\mathbf{F}}^{(i_{d})}_{2,k}\big]_{2,2} = \big(\big[{\mathbf{F}}^{(i_{d})}_{2,k}\big]_{1,1}\big)^{-1}\big[{\mathbf{F}}^{(i_{d})}_{2,k}\big]_{2,1}\big[{\mathbf{F}}^{(i_{d})}_{2,k}\big]_{1,2}.
\end{align}
Hence, the corresponding subproblem can be reformulated as
\begin{align}\label{Opt:PropDisPrecoder}
	\mathop{\min}_{\mathbf{W}, \mathbf{B}} \ &  \sum\nolimits_{k=1}^{K} 2\re\big(\big[{\mathbf{F}}^{(i_{d})}_{2,k}\big]_{1,2}\widehat{\mathbf{h}}^{\mathrm{H}}_{k}\mathbf{B}\mathbf{w}_{k}\big)  + \big[{\mathbf{F}}^{(i_{d})}_{2,k}\big]_{1,1} \\
	&  + \tr\big(\big[{\mathbf{F}}^{(i_{d})}_{2,k}\big]_{2,2}\widehat{\mathbf{h}}^{\mathrm{H}}_{k}\mathbf{B}\mathbf{W}\mathbf{W}^{\mathrm{H}}\mathbf{B}^{\mathrm{H}}\widehat{\mathbf{h}}_{k} + \big[{\mathbf{F}}^{(i_{d})}_{2,k}\big]_{2,2}\omega{\widehat{\sigma}^{2}_{2, k}}\big) \notag \\
	\text{s.t.} \ & \mathrm{{\underline{C}}1:} \tr\big(\big[{\mathbf{F}}^{(i_{d})}_{2,k}\big]_{2,2}\widehat{\mathbf{h}}^{\mathrm{H}}_{k}\mathbf{B}\mathbf{W}\mathbf{W}^{\mathrm{H}}\mathbf{B}^{\mathrm{H}}\widehat{\mathbf{h}}_{k} + \big[{\mathbf{F}}^{(i_{d})}_{2,k}\big]_{2,2}\omega{\widehat{\sigma}^{2}_{2, k}}\big) \notag \\
	&  + 2\re\big(\big[{\mathbf{F}}^{(i_{d})}_{2,k}\big]_{1,2}\widehat{\mathbf{h}}^{\mathrm{H}}_{k}\mathbf{B}\mathbf{w}_{k}\big) + \big[{\mathbf{F}}^{(i_{d})}_{2,k}\big]_{1,1} \leq \widehat{\Gamma}_{k},	\mathrm{\widehat{C}3} \notag
\end{align}
where $\widehat{\Gamma}_{k} = \log(\mathbf{u}^{\mathrm{H}}{\mathbf{D}}_{2,k}^{{(i_{d})},-1}\mathbf{u}) + \tr\big({\mathbf{F}}^{(i_{d})}_{2,k}{\mathbf{D}}^{(i_{d})}_{2,k}\big) - \Gamma_{k}$.
Similar to {Proposition~\ref{Prop:MMSE-SOMP}}, the objective function of \eqref{Opt:PropDisPrecoder} can also be reformulated as an RLS problem, as demonstrated in the following proposition.
\begin{proposition}\label{Prop:MM-SOMP}
	The objective function of \eqref{Opt:PropDisPrecoder} is equivalent to the following form:
	\begin{align}
		f(\mathbf{W}, \mathbf{B}) = \Vert \mathbf{M}^{\frac{1}{2}} -  \mathbf{M}^{-\frac{1}{2}} \mathbf{N}\widehat{\mathbf{H}}\mathbf{B}\mathbf{W} \Vert^{2}_{F} + \zeta_2 \Vert \mathbf{W} \Vert^{2}_{F} 
	\end{align}
	where 
	\begin{align}
		& \widehat{\mathbf{H}} = [\widehat{\mathbf{h}}_{1}, \widehat{\mathbf{h}}_{2}, \cdots, \widehat{\mathbf{h}}_{K}]^{\mathrm{H}} ,\ \zeta_2 = \sum\nolimits_{k=1}^{K}\big[{\mathbf{F}}^{(i_{d})}_{2,k}\big]_{2,2}{\widehat{\sigma}^{2}_{2, k}}/P_{\mathrm{max}}\notag \\
		& {\mathbf{M}} =  \mathrm{diag}\big(\big[{\mathbf{F}}^{(i_{d})}_{2,1}\big]_{1,1}, \cdots, \big[{\mathbf{F}}^{(i_{d})}_{2,K}\big]_{1,1}\big) \notag \\
		& {\mathbf{N}} = - \mathrm{diag}\big(\big[{\mathbf{F}}^{(i_{d})}_{2,1}\big]_{1,2}, \cdots, \big[{\mathbf{F}}^{(i_{d})}_{2,K}\big]_{1,2}\big) 
	\end{align}
\end{proposition}
\begin{IEEEproof}
	See Appendix~\ref{App:D}.
\end{IEEEproof}
Using {Proposition~\ref{Prop:MM-SOMP}} and neglecting QoS constraints, we can transform \eqref{Opt:PropDisPrecoder} into 
\begin{align}\label{Prob:MM-SOMP}
	\ \mathop{\min}_{\widehat{\mathbf{W}}} &  \ \Vert \mathbf{M}^{\frac{1}{2}} -  \mathbf{M}^{-\frac{1}{2}} \mathbf{N}\widehat{\mathbf{H}}\widehat{\mathbf{W}} \Vert^{2}_{F} + \zeta_2 \Vert \widehat{\mathbf{W}} \Vert^{2}_{F} \\
	\qquad \text{s.t.} & \ \mathrm{\underline{C}3:}\Vert \widehat{\mathbf{W}} \Vert_{\mathrm{row}, 0} = N \notag 
\end{align}
where $\widehat{\mathbf{W}}$ denotes the $N$-row-sparse precoder. We solve the problem using RLS-SOMP to obtain $\widehat{\mathbf{W}}^{\star}$. Its nonzero-row indices form $\boldsymbol{\Lambda}_{\mathrm{t}}$, from which the selection matrix $\mathbf{B}$ is constructed. 

Although Proposition~\ref{Prop:MM-SOMP} provides an equivalent least-squares representation of the objective function for fixed model variables, jointly solving the resulting row-sparse problem together with QoS constraints remains intractable. Therefore, we construct a computationally efficient heuristic by temporarily omitting the QoS constraints and applying RLS-SOMP to identify a candidate support set. After recovering the support set, the selection matrix is reconstructed and the power allocation problem in \eqref{Opt:PowerOpt} is solved to evaluate and improve QoS feasibility.
\par
Next, with given $\widehat{\mathbf{W}}^{\star}$ and $\mathbf{B}$, the beamforming direction, denoted by $\{\widetilde{\mathbf{w}}_{k}\}^{K}_{k=1}$ is established. Then, we allocate the power by solving the following problem to satisfying QoS constraints.
\begin{table*}[!t]
	\footnotesize
	\centering
	\caption{Computational Complexity and Runtime Comparison}
	\label{tab:complexity-runtime}
	\setlength{\tabcolsep}{5pt}
	\renewcommand{\arraystretch}{1.25}
	\begin{tabular}{@{}c c cc@{}}
		\toprule
		& \multirow{2}{*}{Computational Complexity}
		& \multicolumn{2}{c}{Runtime (s)} \\
		\cmidrule(lr){3-4}
		& & $W=4\lambda$ & $W=5\lambda$ \\
		\midrule
		
		$\mathrm{Prop.~Conti.}$
		&
		{\scriptsize
			$\begin{aligned}
				\mathcal{O}\Bigg(I_{\mathrm O}\Big[
				KM^{3}
				+I_{1}\big((K+QR)M^{2}L^{2}+M^{3}\big)+I_{2}\big(KN^{2}L^{2}+N^{3}\big)
				+K^{3}N^{3}
				\Big]\Bigg)
			\end{aligned}$}
		& 914.42 & 815.51 \\
		
		$\mathrm{Prop.~Disc.}$
		&
		{\scriptsize
			$\begin{aligned}
				\mathcal{O}\Bigg(I_{\mathrm O}\Bigg\{
				K\Bigg[
				MG_{\mathrm r}(K+RQ)
				+\sum_{m=1}^{M}
				\big(m^{2}(K+RQ)+m^{3}\big)
				\Bigg]+NG_{\mathrm t}K^{2}
				+\sum_{n=1}^{N}
				\big(n^{2}K^{2}+n^{3}\big)
				+K^{3.5}
				\Bigg\}\Bigg)
			\end{aligned}$}
		& 0.097 & 0.124 \\
		
		\bottomrule
	\end{tabular}
	\vspace{-3mm}
\end{table*}
\setlength{\textfloatsep}{5pt}
\begin{align}\label{Opt:PowerOpt}
	\mathop{\max}_{\mathbf{p}} \ &  \sum\nolimits_{k=1}^{K} \log\Bigg(1 + \frac{p_k \vert \widehat{\mathbf{h}}^{\mathrm{H}}_{k}\mathbf{B}\widetilde{\mathbf{w}}_{k} \vert^{2}}{\sum_{i \neq k}^{K}p_i \vert \widehat{\mathbf{h}}^{\mathrm{H}}_{k}\mathbf{B}\widetilde{\mathbf{w}}_{i} \vert^{2} + \omega{\widehat{\sigma}^{2}_{2, k}}}\Bigg) \\
	\text{s.t.} \ & \log\Bigg(1 + \frac{p_k \vert \widehat{\mathbf{h}}^{\mathrm{H}}_{k}\mathbf{B}\widetilde{\mathbf{w}}_{k} \vert^{2}}{\sum_{i \neq k}^{K}p_i \vert \widehat{\mathbf{h}}^{\mathrm{H}}_{k}\mathbf{B}\widetilde{\mathbf{w}}_{i} \vert^{2} + \omega{\widehat{\sigma}^{2}_{2, k}}}\Bigg) \geq \Gamma_{k} \notag \\
	& \sum^{K}_{k=1} p_k \leq P_{\max}, p_k \geq 0, \forall k \notag
\end{align}
where $\mathbf{p} = [p_{1},p_{2},\cdots,p_{K}]^{\mathrm{T}}$ denote the power allocation vector for users. To solve problem \eqref{Opt:PowerOpt}, we adopt the Lagrangian dual transformation and quadratic transformation proposed in \cite{8314727} and the problem can be equivalently written as
\begin{align}\label{eq:opt_power}
	\max_{\mathbf{p}, \bar{\mathbf{y}}, \mathbf{z}} \ & \sum_{k=1}^{K} R'_{k}(\mathbf{p}, y_{k}, z_{k}) \\
	\text{s.t.} \ & R'_{k}(\mathbf{p}, \bar{y}_{k}, z_{k}) \geq \Gamma_{k}, \forall k \notag 
\end{align}
where $\bar{\mathbf{y}} = [\bar{y}_{1}, \bar{y}_{2}, \dots,\bar{y}_{K}]^{\mathrm{T}}$ and $\mathbf{z} = [z_{1}, z_{2}, \dots,z_{K}]^{\mathrm{T}}$ are auxiliary variables for the SINR terms, 
\begin{align}
	R'_{k} = & \log(1+\bar{y}_{k}) - \bar{y}_{k} + 2z_{k}\sqrt{(1+\bar{y}_{k})p_{k}}\big\vert\widehat{\mathbf{h}}^{\mathrm{H}}_{k}\mathbf{B}\widetilde{\mathbf{w}}_{k}\big\vert \notag \\ 
	& - z^{2}_{k}\Bigg(\sum_{i=1}^{K}p_{i}\vert\widehat{\mathbf{h}}^{\mathrm{H}}_{k}\mathbf{B}\widetilde{\mathbf{w}}_{i}\vert^2+{\| \mathbf{p}\|_{1}}\widehat{\sigma}^{2}_{2, k}/P_{\max}\Bigg).
\end{align}
The optimal $\bar{\mathbf{y}}$ and $\mathbf{z}$ are given by
\begin{align}
	\bar{y}_{k} &= \frac{p_k \vert \widehat{\mathbf{h}}^{\mathrm{H}}_{k}\mathbf{B}\widetilde{\mathbf{w}}_{k} \vert^{2}}{\sum_{i \neq k}^{K}p_i \vert \widehat{\mathbf{h}}^{\mathrm{H}}_{k}\mathbf{B}\widetilde{\mathbf{w}}_{i} \vert^{2} + {\| \mathbf{p}\|_{1}}{\widehat{\sigma}^{2}_{2, k}}/{P_{\max}}}\label{eq:opt_y} \\
	z_{k} & = \frac{\sqrt{(1+\bar y_{k})p_{k}}\big\vert\widehat{\mathbf{h}}^{\mathrm{H}}_{k}\mathbf{B}\widetilde{\mathbf{w}}_{k}\big\vert}{\sum_{i=1}^{K}p_{i}\vert\widehat{\mathbf{h}}^{\mathrm{H}}_{k}\mathbf{B}\widetilde{\mathbf{w}}_{i}\vert^2+{\| \mathbf{p}\|_{1}}{\widehat{\sigma}^{2}_{2, k}}/{P_{\max}}}\label{eq:opt_z}.
\end{align}
Then, we derive a semi-closed form solution for \eqref{eq:opt_power} and propose a computationally efficient subgradient method to search for the Lagrangian multipliers. Formally, the Lagrangian of \eqref{eq:opt_power} is given by
\begin{align}
	\mathcal{L}(\mathbf{p}, \bm{\mu}) = \sum_{k=1}^{K}(1+\mu_{k})R'_{k}(\mathbf{p})
\end{align}
where $\bm{\mu}=[\mu_{1}, \mu_{2}, \dots, \mu_{K}]^{\mathrm{T}}$ is the vector composed of multiple non-negative Lagrange multipliers. By taking the first-order derivative of $\mathcal{L}$ over $\mathbf{p}$ and setting it to zero, we have
\begin{align}\label{eq:opy_p}
	p^{\star}_i
	=
	\left[
	\frac{(1+\mu_i)z_i\sqrt{1+y_i}\,\vert \widehat{\mathbf{h}}^{\mathrm{H}}_{i}\mathbf{B}\widetilde{\mathbf{w}}_{i} \vert}{\displaystyle\sum_{k=1}^{K}(1+\mu_k)z_k^2\left(\vert\widehat{\mathbf{h}}^{\mathrm{H}}_{k}\mathbf{B}\widetilde{\mathbf{w}}_{i}\vert^2+{\widehat{\sigma}_{2,k}^2}/{P_{\max}}\right)}\right]_+^2.
\end{align}
where $[x]_+\triangleq\max\{x,0\}.$ Then, we can obtain the optimal dual variable $\{\mu^{\star}_{k}\}^{K}_{k=1}$ by the subgradient method. The update rule is
\begin{align}
	\mu^{(t+1)}_{k} = \Big[\mu^{(t)}_{k} + \eta_{t}(\Gamma_{k} - R'_{k}{}^{(t)}(\mathbf{p}^{\star}))\Big]_{+}.
\end{align}
where $\eta_{t}$ is step size, satisfying $\sum_{t=0}^{\infty}\eta_t=\infty$ and $\sum_{t=0}^{\infty}\eta_t^2<\infty.$ The solution process for problem~\eqref{Opt:PowerOpt} is concluded in Algorithm~\ref{Alg:PC}.
\begin{algorithm}[t]
	\normalsize
	\caption{The Proposed Algorithm for Problem~\eqref{Opt:PowerOpt}}\label{Alg:PC}
	\begin{algorithmic}[1]
		\REQUIRE $P_{\max}$, $\{\Gamma_{k}, \widehat{\mathbf{h}}_{k}, \widetilde{\mathbf{w}}_{k}, \widehat{\sigma}^{2}_{2,k}\}$, and $\mathbf{B}$;
		\ENSURE $\mathbf p$;
		\STATE \textbf{Initialize} $p_{k} = P_{\max}/K$ and $\mu_{k}=0$;
		\REPEAT
		\STATE Update $\mathbf{y}$ with \eqref{eq:opt_y}, $\mathbf{z}$ with \eqref{eq:opt_z};
		\REPEAT
		\STATE $t \Leftarrow 0$;
		\FOR{$i=1$ to $K$}
		\STATE Update $p^{\star}_{i}$ according to \eqref{eq:opy_p};
		\ENDFOR
		\STATE Update QoS multipliers by $\mu^{(t+1)}_{k} = \Big[\mu^{(t)}_{k} + \eta_{t}(\Gamma_{k} - R'_{k}{}^{(t)}(\mathbf{p}^{\star}))\Big]_{+}$;
		\STATE $t \Leftarrow t+1$;
		\UNTIL{$\sum_{k=1}^{K}|\mu_{k}^{(t+1)}-\mu_{k}^{(t)}|\leq\epsilon$}
		\UNTIL{the objective values of \eqref{eq:opt_power} converges or the iteration  number reaches the maximum.}
		\STATE Normalize the power vector, i.e., $\mathbf{p} \Leftarrow {P_{\max}}\cdot\mathbf{p}/{\|\mathbf{p}\|_{1}}$;
	\end{algorithmic}
\end{algorithm}
\subsection{Complexity Analysis}
\begin{algorithm}[!t]
	\normalsize
	\caption{BCD Algorithm for Problem \eqref{Opt:PropDis}}
	\label{Alg:BCD}
	\begin{algorithmic}[1]
		\REQUIRE $\widehat{\mathbf{H}}_{k}$ and $\widehat{\mathbf{g}}_{r,k}$;
		\ENSURE $\mathbf{W}$, $\mathbf{v}_{k}$, $\mathbf{B}$, and $\mathbf{C}_{k}$;
		\STATE $i_{d} \Leftarrow 0$;
		\REPEAT
		\STATE Update $\mathbf{v}_{k}$ and $\mathbf{C}_{k}$ by solving \eqref{Prob:MMSE-SOMP} via RLS-SOMP;
		\STATE Calculate ${\mathbf{D}}^{(i_{d})}_{2,k}$ and ${\mathbf{F}}^{(i_{d})}_{2,k}$;
		\STATE Update $\mathbf{W}^{\star}$ and $\mathbf{B}$ by solving \eqref{Prob:MM-SOMP} via RLS-SOMP;
		\STATE Normalize the beamforming direction for each user, i.e., $\widetilde{\mathbf{w}}_{k} \Leftarrow {\mathbf{w}}^{\star}_{k}/\Vert {\mathbf{w}}^{\star}_{k}\Vert_{2}$, $\forall k$;
		\STATE Solve \eqref{Opt:PowerOpt} using $\{\widetilde{\mathbf{w}}_{k}\}_{k=1}^{K}$ to obtain $\mathbf p$;
		\STATE Set ${\mathbf{w}}_{k} \Leftarrow \sqrt{p_k}\widetilde{\mathbf{w}}_{k}$, $\forall k$;
		\STATE $i_{d} \Leftarrow i_{d} + 1$;
		\UNTIL{the objective values of \eqref{Opt:PropDis} converges or the iteration  number reaches the maximum.}
	\end{algorithmic}
\end{algorithm}
\setlength{\textfloatsep}{5pt}
To better understand the heuristic BCD framework, we summarize it in {Algorithm~\ref{Alg:BCD}}. Due to the greedy support-selection operation of RLS-SOMP, the proposed algorithm does not exactly minimize each block subproblem of \eqref{Opt:PropDis}. Therefore, the classical convergence guarantees of exact BCD do not apply, and neither global optimality nor convergence to a stationary point of \eqref{Opt:PropDis} is claimed. The required number of operations for updating $\mathbf{v}_{k}$ and $\mathbf{C}_{k}$ can be expressed as $C_{\mathrm{r}} = MG_{\mathrm{r}}(K+RQ) + \sum_{m=1}^{M}(m^{2}(K+RQ) + m^{3})$. When updating $\mathbf{W}$ and $\mathbf{B}$ via the SOMP method, the operations can be represented as $C_{\mathrm{t}, 1} = NG_{\mathrm{t}}K^{2} + \sum_{n=1}^{N}(n^{2}K^{2} + n^{3})$. Additionally, the power allocation step requires $C_{\mathrm{t}, 2} = \mathcal{O}(K^{3.5})$ operations. Accordingly, the total computational complexity of the proposed algorithm is given by $\mathcal{O}_{\mathrm{BCD}} = \mathcal{O}(\max\{K C_{\mathrm{r}}, C_{\mathrm{t}, 1} + C_{\mathrm{t}, 2}\})$. 

Moreover, given the large number of iterations $I_{1}$ and $I_{2}$ demanded by the AO algorithm for position updates, the BCD-based optimization framework has much lower complexity than that of AO and does not rely on any toolbox, which is beneficial to practical implementation. The computational complexities and runtimes of the proposed algorithms are summarized in Table~\ref{tab:complexity-runtime}, {with the parameters defined in Section~\ref{sec:sim}. All algorithms are implemented on a PC with Windows 10 and Intel Core i7-14700 CPU @ 5.60GHz.}
\section{Simulation Results}\label{sec:sim}
In this section, numerical results are provided to evaluate the validity and superiority of our proposed scheme and algorithm.
\begin{table}[t]
	\caption{Simulation Parameters}
	\label{Tab:par}
	\centering
	
	\begin{tabular}{cc}
		\toprule
		\textbf{Parameters} & \textbf{Values}  \\
		\midrule
		\rowcolor{white!92!cyan}
		Number of MAs at the BS $N$ &  $16$     \\
		Number of MAs at the user $M$ & $9$     \\
		\rowcolor{white!92!cyan}
		Number of channel paths $L$ & $8$ \\
		Carrier wavelength $\lambda$ & $0.1$~m~\cite{F-WMMSE} \\ %
		\rowcolor{white!92!cyan}
		Length of the sides of moving region $W_{\mathrm{t}}, W_{\mathrm{r}}$ & $4\lambda$ \\
		Minimum inter-MA distance $D$ & $\lambda/2$ \\
		\rowcolor{white!92!cyan}
		Channel gain at the reference distance $\rho$ & $-30$~dB \\
		Path loss exponent $\alpha$ & $2.6$~\cite{10296481} \\
		\rowcolor{white!92!cyan}
		White noise $\sigma^{2}_{k}$ & $-70$~dBm~\cite{9264659} \\ %
		Maximum power at BS $P_{\mathrm{max}}$ & $10$~dBm~\cite{11403971}  \\ %
		\rowcolor{white!92!cyan}
		Minimum rate threshold $\Gamma_{k}$ & $1$~bps/Hz~\cite{9264659} \\ %
		\bottomrule
	\end{tabular}
\end{table}
\begin{figure}[!t]
	\centering
	\includegraphics[width=0.36\textwidth]{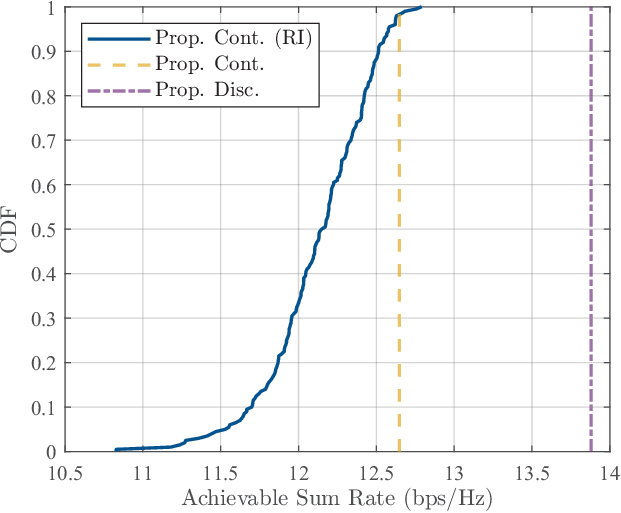}
	\captionsetup{font=footnotesize, name={Fig.}, labelsep=period} 
	\caption{\raggedright CDF of the sum rate obtained by the continuous-position algorithm with random initialization (RI). The results reveal that AO is sensitive to the initial antenna positions, while the proposed initialization strategy achieves excellent performance.}
	\label{res:CDF}
	\vspace{-5mm}
\end{figure}
\subsection{Simulation Setup}
We consider a scenario of $K=3$ users and $R=2$ jammers, where the BS is deployed at $(0, 0, 30)$~m and the two jammers are located at $(4, -13.5, 48.5)$~m and $(-0.5, -8.5, 66)$~m, respectively. The three users are distributed at the directions $\{(\theta, \phi) \lvert (-40^{\circ}, 30^{\circ}), (-30^{\circ}, 50^{\circ}), (-20^{\circ}, 70^{\circ})\}$ of BS with a radius of 100 m. The main simulation parameters are detailed in {Table~\ref{Tab:par}}, unless stated otherwise. In addition, we employ a geometric channel model, assuming that the number of propagation paths for each channel is the same, i.e., $L^{\mathrm{t}}_{k} = L^{\mathrm{r}}_{k} = L^{\mathrm{t}}_{r,k}= L^{\mathrm{r}}_{r,k} = L$. The path response matrix is assumed to be diagonal with $\bm{\Sigma}\left[1,1\right] \sim \mathcal{CN}\left(0,\rho d^{-\alpha}\right)$ and $\bm{\Sigma}\left[l,l\right] \sim \mathcal{CN}\left(0,\rho d^{-\alpha}/(L-1)\right)$ for $l=2,3,\dots,L$. Here, $\rho d^{-\alpha}$ represents the expected channel gain. CSI uncertainty is set $\Delta = \theta_{U} - \theta_{L} = \phi_{U} - \phi_{{L}}$ \cite{10296481}. Moreover, {signal-to-jamming plus-noise-ratio (SJNR) is defined as $\mathrm{SJNR} = [P_{\mathrm{max}}/\sum_{r=1}^{R}p_{\mathrm{J},r}]_{[\mathrm{dB}]}$~\cite{10296481}}. Here, we compare our proposed scheme and algorithms with the following benchmarks:
\begin{itemize}
	\item \textbf{FPA}: The antenna arrays at the BS and the users are fixed in position with a spacing of $\lambda/2$.
	\item \textbf{Random Position Antenna (RPA)}: The antennas at the BS and the users are MAs with random positions.
	\item \textbf{Continuous Position BCA (CP-BCA)} \cite{JSACMA}: Fractional programming and the MM algorithm were employed to optimize beamforming and MA positions in an alternating fashion.
	\item \textbf{Discrete Position FWMMSE (DP-FWMMSE)} \cite{F-WMMSE}: This method reformulates the WMMSE subproblems into a regularized sparse optimization framework, utilizing RLS-SOMP to achieve joint beamforming and discrete antenna position optimization.
\end{itemize}
\begin{figure}[!t]
	\centering
	{\includegraphics[width=0.36\textwidth]{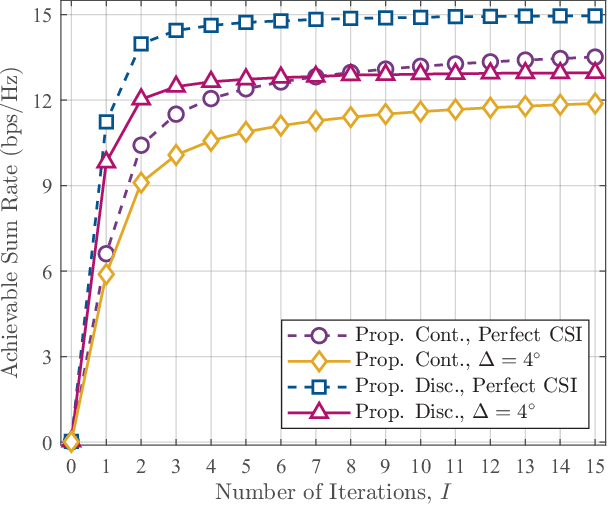}}  \vspace{-0mm}
	\captionsetup{font=footnotesize, name={Fig.}, labelsep=period}  
	\caption{\raggedright Achievable sum rate versus number of iterations, given $\Delta = 4^{\circ}$, and SJNR $=-20$~dB. The results indicate that proposed discrete design achieves faster convergence while maintaining superior system performance.}
	\label{res:conve}
\end{figure}
\begin{figure}[!t]
	\centering
	{\includegraphics[width=0.36\textwidth]{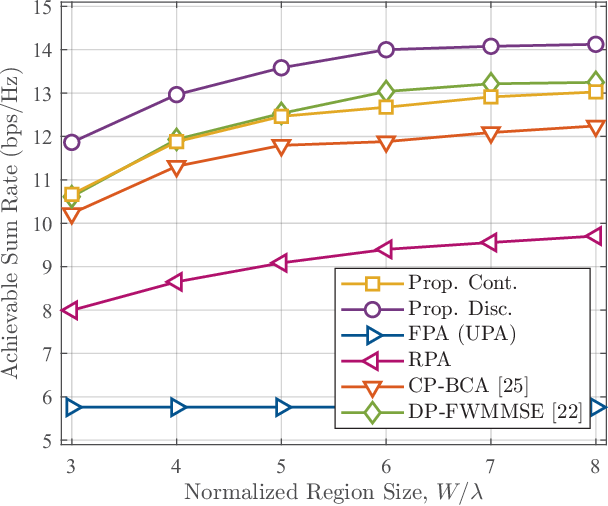}}  \vspace{-0mm}
	\captionsetup{font=footnotesize, name={Fig.}, labelsep=period}  
	\caption{\raggedright Achievable sum rate versus normalized region size $W/\lambda$, given $\Delta = 4^{\circ}$, and SJNR $=-20$~dB. The results reveal that the proposed discrete design consistently outperforms the baseline schemes across all moving sizes.}
	\label{res:size}
\end{figure}
\subsection{Initialization Analysis}
In Fig.~\ref{res:CDF}, we examine the initialization sensitivity of the continuous-position AO algorithm. The achievable sum rates obtained with random initialization span approximately $10.8$--$12.8$~bps/Hz, confirming that the stationary point reached by AO depends strongly on the initial antenna positions. By contrast, the proposed continuous-position initialization achieves approximately $12.65$~bps/Hz and lies near the upper tail of the distribution, outperforming the vast majority of the random trials. Thus, the adopted initialization provides a more favorable and repeatable starting point than arbitrary random initialization.
\subsection{Performance Analysis}
\begin{figure*}[!t]
	\centering
	\subfloat[]{  
		\includegraphics[width=0.33\textwidth]{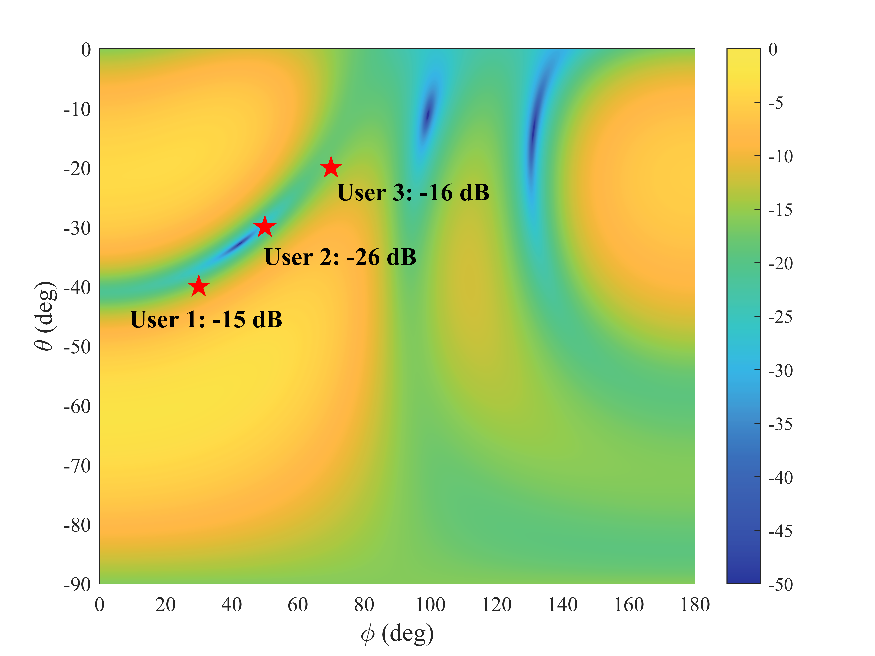}
		\label{fig:sub_a}
	}
	\subfloat[]{  
		\includegraphics[width=0.33\textwidth]{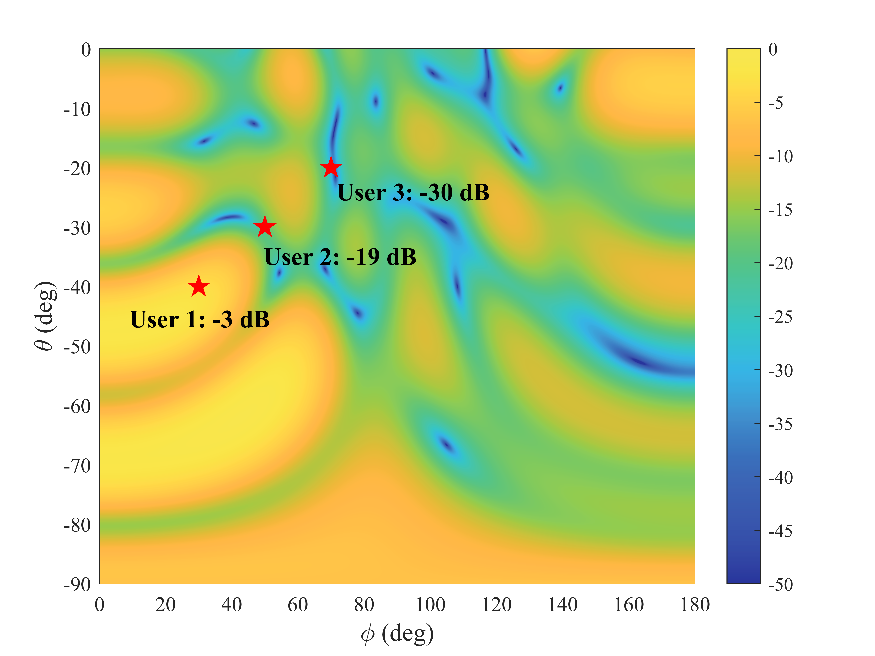}
		\label{fig:sub_b}
	}
	\subfloat[]{  
		\includegraphics[width=0.33\textwidth]{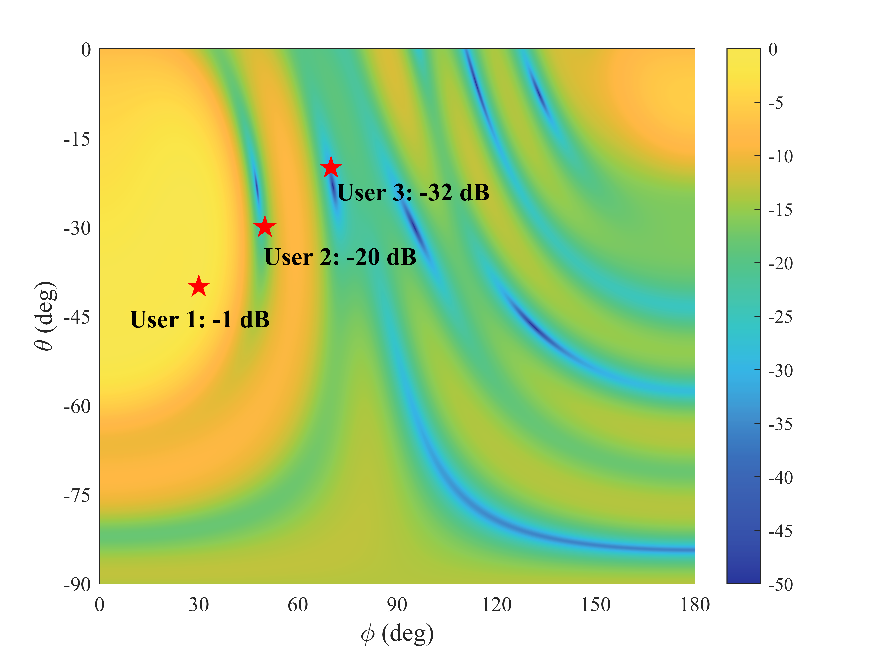}
		\label{fig:sub_c}
	}
	\vspace{-4mm}
	\\
	\subfloat[]{  
		\includegraphics[width=0.33\textwidth]{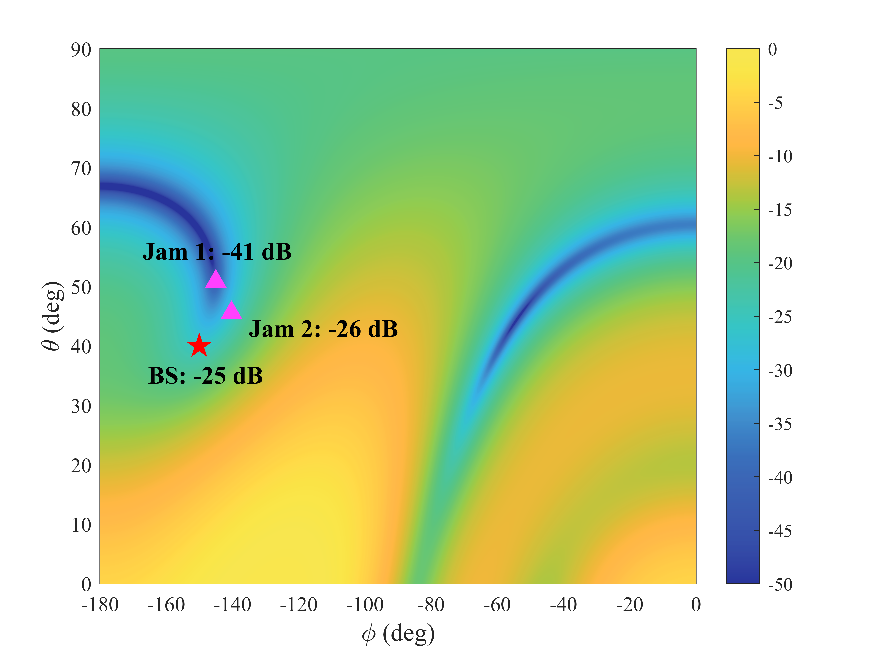}
		\label{fig:sub_d}
	}
	\subfloat[]{  
		\includegraphics[width=0.33\textwidth]{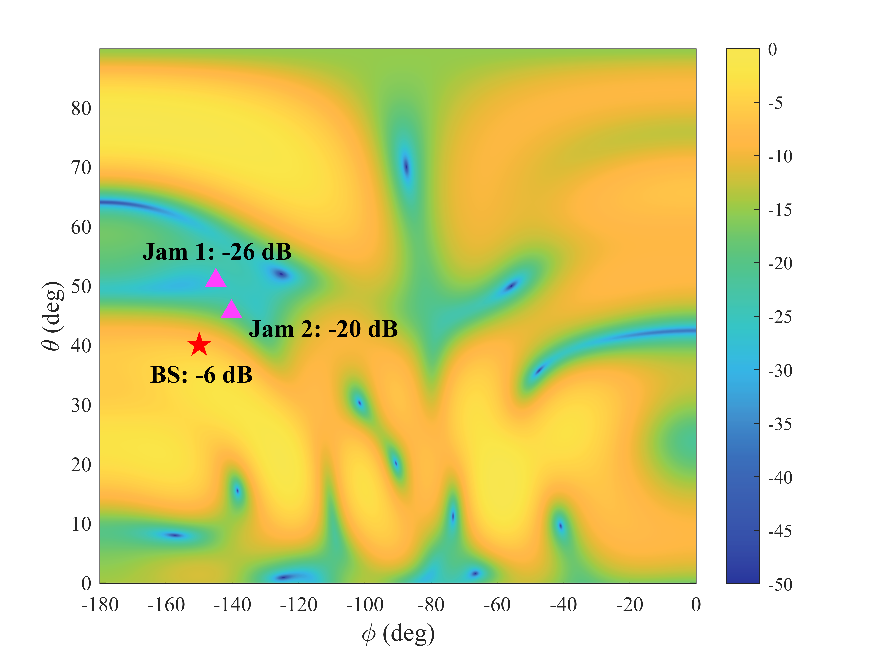}
		\label{fig:sub_e}
	}
	\subfloat[]{  
		\includegraphics[width=0.33\textwidth]{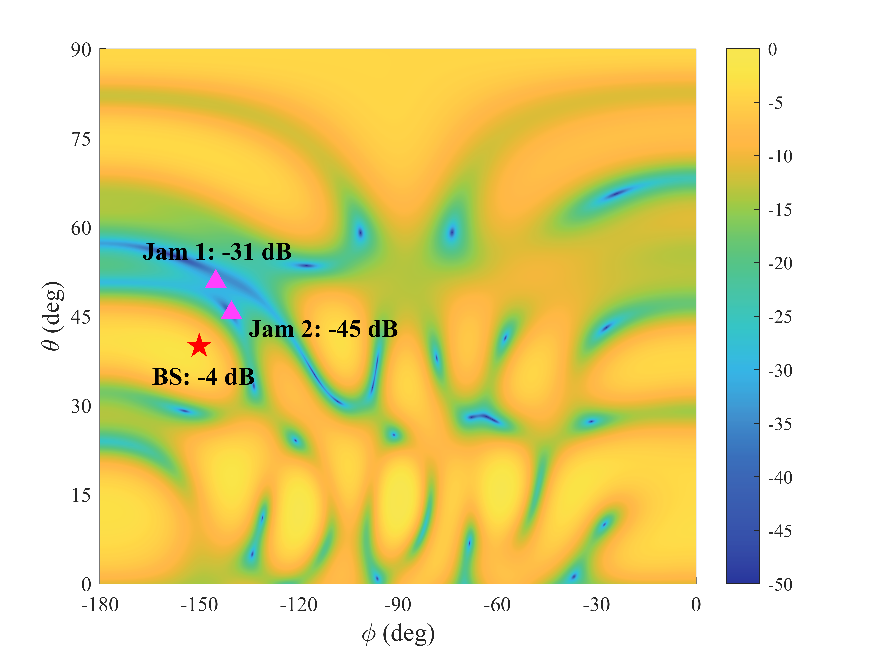}
		\label{fig:sub_f}
	}
	\captionsetup{font=footnotesize, name={Fig.}, labelsep=period}  
	\caption[t]{\raggedright Beampatterns obtained by the compared positioning algorithms for transmission to user $1$ (unit: dB): (a) transmit FPA; (b) transmit continuous-position AO; (c) transmit discrete-position BCD; (d) receive FPA; (e) receive continuous-position AO; (f) receive discrete-position BCD, with $\Delta=4^{\circ}$ and SJNR $=-20$~dB. The results demonstrate that FAS provides enhanced beamforming flexibility over FPA by exploiting the additional DoF. Furthermore, the proposed discrete design enhances the main-lobe gain while creating deep nulls toward interfering directions, thereby achieving higher achievable rates.}
	\label{res:beam}
	\vspace{-3mm}
\end{figure*}
\begin{figure}[!t]
	\begin{center}		{\includegraphics[width=0.36\textwidth]{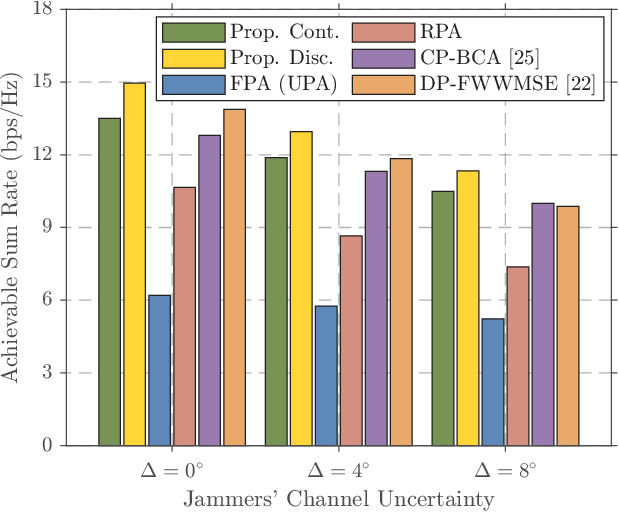}}  \vspace{-0mm}
		\captionsetup{font=footnotesize, name={Fig.}, labelsep=period}  
		\caption[t]{\raggedright Achievable sum rate versus jammers' channel uncertainty $\Delta$, given SJNR $=-20$~dB. The results reveal that the proposed discrete design achieves the best performance under various angle errors, while the continuous design shows better robustness.}
		\label{res:Error}
	\end{center}
\end{figure}
\begin{figure}[!t]
	\begin{center}		{\includegraphics[width=0.36\textwidth]{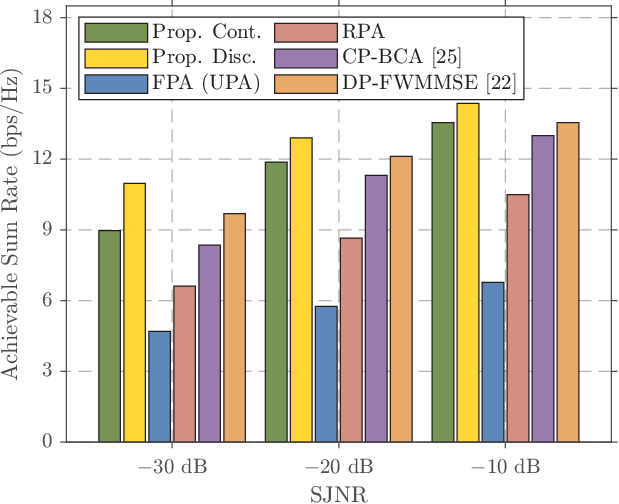}}  \vspace{-0mm}
		\captionsetup{font=footnotesize, name={Fig.}, labelsep=period}  
		\caption[t]{\raggedright Achievable sum rate versus SJNR, given $\Delta = 4^{\circ}$. The results show that our schemes achieve significantly higher rates than the existing schemes at low SJNRs.}
		\label{res:SJNR}
	\end{center}
\end{figure}
\textbf{\textit{Q1: How do the proposed discrete- and continuous-position algorithms compare?}}

\textit{\textbf{Observation 1}}: \textit{The discrete BCD algorithm achieves higher sum rates than the continuous AO algorithm and can efficiently utilize the extended spatial DoF. (cf. Fig.~\ref{res:conve} and Fig.~\ref{res:size})}

Fig.~\ref{res:conve} illustrates the convergence behavior of the proposed schemes under different configurations. The proposed optimization framework achieves stable convergence under both continuous and discrete implementations across various CSI conditions, as indicated by the stabilized mean value. It is observed that, under the adopted optimization strategy, the discrete scheme yields better overall performance than the continuous scheme. {The gain of continuous AO scheme over FPA demonstrates the benefit of reconfigurable antenna positions, while the further gain of the discrete heuristic BCD highlights the benefit of jointly updating antenna positions and beamformer within each transmit or receive block.}

Fig.~\ref{res:size} demonstrates the impact of the physical size of the movable region on the sum achievable rate. It can be observed that the achievable sum rate gradually increases and then converges as the region size expands. This is because enlarging the movable region offers more DoFs for position optimization of the MAs, thereby improving anti-jamming performance. Specifically, even with a normalized region size of 3, the proposed discrete design achieves approximately 106\% performance improvement compared to the traditional FPA scheme, without requiring additional RF chains. Since RF chains are typically expensive, the proposed design substantially reduces the hardware burden, offering a highly efficient solution for anti-jamming communication. Nevertheless, it should be noted that the performance improvement resulting from the expanded region size is limited, and the increase in the sum-rate slows down as $W/\lambda$ increases. This indicates that in practical system design, a trade-off should be made between area size and performance gains to achieve the optimal balance between cost and efficiency.
\subsection{Beampattern Analysis}
\textbf{\textit{Q2: What beam patterns are produced by the two positioning algorithms?}}

\textit{\textbf{Observation 2}}: \textit{Compared with FPA and proposed continuous FA design, the proposed discrete FA design forms more focused beams toward desired directions while effectively suppressing multi-user interference and jamming signals through more precise spatial responses. (cf. Fig.~\ref{res:beam})}

To intuitively interpret this underlying mechanism, Fig.~\ref{res:beam} illustrates the normalized transmit and receive beampattern of different architectures and methods. For the transmit side, Fig.~\ref{res:beam}(a)–(c) show that the proposed discrete design generates a more concentrated main lobe pointing toward the intended user direction, while producing lower sidelobe levels compared with FPA and the proposed continuous design. Using the transmit beampattern of $\mathbf{w}_{1}$ as an example, the proposed discrete design attains approximately $-1$~dB at the direction of user~1, whereas deep nulls with depths of $-20$~dB and $-32$~dB are generated toward the remaining users. 
This indicates that the proposed discrete design is able to more effectively exploit spatial DoFs through joint optimization.

On the receive side, as shown in Fig.~\ref{res:beam}(d)–(f), proposed discrete design exhibits superior jamming suppression capability. It can be observed that the proposed discrete design can accurately generate the nulls towards the jammers’ regions, and simultaneously align the mainlobes with the desired target, even under the angular uncertainty. In contrast, the FPA scheme suffers from distorted beam patterns and higher sidelobe leakage, which may degrade interference mitigation performance. In particular, the received SINR of the proposed discrete design at the direction of the BS is approximately $-4$~dB, whereas those of the FPA and the proposed continuous design are about $-25$~dB and $-6$~dB, respectively. The above observations confirm that our proposed discrete design can fully exploit DoFs in the spatial domains to realize anti-jamming communications.
\subsection{Robustness Analysis}
\textbf{\textit{Q3: How do the algorithms behave under strong jamming and angular uncertainty?}}

\textit{\textbf{Observation 3}}: \textit{The proposed discrete design maintains a significant performance advantage under severe jamming and imperfect CSI. However, bounded by finite position grids, its degradation rate against increasing channel uncertainty is slightly steeper than that of the highly robust proposed continuous design. (cf. Fig.~\ref{res:Error} and Fig.~\ref{res:SJNR})}

Fig.~\ref{res:Error} shows the achievable sum rate versus the jamming channel uncertainty $\Delta$. It can be seen that the sum achievable rate decreases with $\Delta$, and the proposed discrete design outperforms the other state-of-the-art benchmarks significantly across all channel uncertainties. It is worth noting that the discrete design is more sensitive to $\Delta$, while the proposed continuous design maintains strong robustness as $\Delta$ increases. This is mainly because the discrete design is limited by the finite DoFs in antenna position optimization. When jammer’s angular CSI error increases, it becomes difficult to select the optimal antenna layout, thus limiting its performance. Nevertheless, the proposed discrete design still exhibits stronger robustness than DP-FWMMSE, and achieves performance improvements of approximately 9\% and 125\% compared to the proposed continuous design and FPA, respectively.

The relationship between the achievable sum rate and SJNR is plotted in Fig.~\ref{res:SJNR}. It can be seen that the lower SJNR is, the lower achievable sum rate will be. Notably, even at an SJNR of $-30$~dB, the proposed discrete design still achieves performance gains of approximately 13.3\% and 22.4\% compared with the state-of-the-art discrete design and the proposed continuous design, respectively. The above results imply that our proposed optimization framework is more robust under powerful jamming attacks.

\section{Conclusion}\label{sec:conl}
This paper developed two optimization frameworks for continuous- and discrete-position designs in MA-assisted anti-jamming communications. Specifically, the continuous-position design framework employs AO to iteratively update the antenna positions and beamforming variables, whereas the discrete-position design framework adopts a sparse-recovery-based heuristic to jointly identify candidate antenna positions and the associated beamforming coefficients. The numerical results reveal two notable observations. First of all, under the considered simulation settings, the proposed discrete-position algorithm achieves higher sum rates than the adopted AO-based continuous-position algorithm, indicating that improved optimization tractability can compensate for the loss of spatial resolution caused by position discretization. Second, the proposed discrete-position algorithm maintains its performance advantage over the considered benchmarks under strong jamming and angular CSI uncertainty, although its degradation with increasing uncertainty is slightly more pronounced than that of the continuous-position algorithm because of the finite grid resolution. These results suggest that the sparse-recovery-based joint treatment of antenna positioning and beamforming may alleviate the sensitivity of conventional AO methods to unfavorable local stationary points and enable more effective exploitation of the spatial DoF provided by MA. Overall, the proposed discrete-position heuristic achieves substantial sum-rate improvements over the considered state-of-the-art benchmarks, demonstrating its effectiveness as a low-complexity solution for MA-assisted anti-jamming communications.

\appendices
\fnbelowfloat
\numberwithin{equation}{section}
\renewcommand{\theequation}{\thesection.\arabic{equation}}
\renewcommand{\theHequation}{appendix.\thesection.\arabic{equation}}

\section{Proof of Proposition~\texorpdfstring{\ref{Prop:SCA}}{1}}\label{App:A}
Define $\overline{\mathbf a}=[|a_1|,\ldots,|a_M|]^{\mathrm T}$, $\kappa_a=\sum_{m=1}^{M}|a_m|$, and
\begin{equation}
\mathbf D_a=\kappa_a\operatorname{diag}(\overline{\mathbf a}), \quad\mathbf A_a=\overline{\mathbf a}\,\overline{\mathbf a}^{\mathrm T}. \label{eq:appA_compact_matrices}
\end{equation}
Then, $q(\mathbf r_k)$ can be expressed as \eqref{eq:appA_direct_expansion}, where
\begin{figure*}[!b]
\noindent\rule{\textwidth}{0.4pt}
\begin{align}
q(\mathbf r_k) &\triangleq\left|\mathbf b^{\mathrm H}\mathbf F_k(\mathbf R_k)\mathbf a\right|^2 =\sum_{l_1=1}^{L_k^{\mathrm r}}\sum_{l_2=1}^{L_k^{\mathrm r}} \sum_{m_1=1}^{M}\sum_{m_2=1}^{M} \re\!\left\{b_{l_1}^{*}b_{l_2}a_{m_1}a_{m_2}^{*} e^{j\tau_{l_1,l_2,m_1,m_2}(\mathbf r_k)}\right\}\notag\\
&=\sum_{l_1=1}^{L_k^{\mathrm r}}\sum_{l_2=1}^{L_k^{\mathrm r}} |b_{l_1}b_{l_2}| \sum_{m_1=1}^{M}\sum_{m_2=1}^{M}|a_{m_1}a_{m_2}| \cos f_{l_1,l_2,m_1,m_2}(\mathbf r_k). \label{eq:appA_direct_expansion}
\end{align}
\end{figure*}
\begin{align}
\tau_{l_1,l_2,m_1,m_2}(\mathbf r_k) &=2\pi\left(\mathbf p_{k,l_1}^{\mathrm T}\mathbf r_{k,m_1} -\mathbf p_{k,l_2}^{\mathrm T}\mathbf r_{k,m_2}\right)/\lambda\notag\\
\phi_{l_1,l_2,m_1,m_2} &=\angle(b_{l_1}^{*}b_{l_2}a_{m_1}a_{m_2}^{*})\label{eq:appA_direct_phase}\\
f_{l_1,l_2,m_1,m_2}(\mathbf r_k) &=\tau_{l_1,l_2,m_1,m_2}(\mathbf r_k) +\phi_{l_1,l_2,m_1,m_2}.\notag
\end{align}
Here, $\mathbf p_{k,l}=[\varphi_{k,l}^{\mathrm r}, \vartheta_{k,l}^{\mathrm r}]^{\mathrm T}$. For any $f,f_0\in\mathbb R$, we have
\begin{align}
\cos f&\leq\cos f_0-\sin f_0(f-f_0)+\tfrac12(f-f_0)^2 \label{eq:appA_upper_cos}\\
\cos f&\geq\cos f_0-\sin f_0(f-f_0)-\tfrac12(f-f_0)^2. \label{eq:appA_lower_cos}
\end{align}
Let $\mathbf r_{k,0}$ denote the current feasible antenna position, and set $\tau_0=\tau_{l_1,l_2,m_1,m_2}(\mathbf r_{k,0})$ and $f_0=f_{l_1,l_2,m_1,m_2}(\mathbf r_{k,0})$. Define
\begin{align}
s_{m_1,m_2}^{\pm}&=\sin f_0\mathbin{\pm}\tau_0\notag\\
c_{m_1,m_2}^{\pm}&=\cos f_0+\tau_0\sin f_0 \mathbin{\pm}\tfrac12\tau_0^2. \label{eq:appA_compact_coefficients}
\end{align}
Here, $(l_1,l_2)$ is omitted for brevity, and the $+$ and $-$ signs correspond to the upper and lower bounds, respectively. Then,
\begin{equation}
\mathbf U_{l_1,l_2}^{\pm} =\pm\frac{2\pi^2}{\lambda^2} \begin{bmatrix} \mathbf A_{l_1,l_2}^{xx}&\mathbf A_{l_1,l_2}^{xy}\\ (\mathbf A_{l_1,l_2}^{xy})^{\mathrm T}&\mathbf A_{l_1,l_2}^{yy} \end{bmatrix} \label{eq:appA_direct_U}
\end{equation}
where
\begin{equation}
\begin{aligned} \mathbf A_{l_1,l_2}^{xx} &=\big((\varphi_{k,l_1}^{\mathrm r})^2 +(\varphi_{k,l_2}^{\mathrm r})^2\big)\mathbf D_a -2\varphi_{k,l_1}^{\mathrm r}\varphi_{k,l_2}^{\mathrm r}\mathbf A_a\\
\mathbf A_{l_1,l_2}^{yy} &=\big((\vartheta_{k,l_1}^{\mathrm r})^2 +(\vartheta_{k,l_2}^{\mathrm r})^2\big)\mathbf D_a -2\vartheta_{k,l_1}^{\mathrm r}\vartheta_{k,l_2}^{\mathrm r}\mathbf A_a\\
\mathbf A_{l_1,l_2}^{xy} &=\big(\varphi_{k,l_1}^{\mathrm r}\vartheta_{k,l_1}^{\mathrm r} +\varphi_{k,l_2}^{\mathrm r}\vartheta_{k,l_2}^{\mathrm r}\big)\mathbf D_a\\
&\qquad\qquad\quad-\big(\varphi_{k,l_1}^{\mathrm r}\vartheta_{k,l_2}^{\mathrm r} +\varphi_{k,l_2}^{\mathrm r}\vartheta_{k,l_1}^{\mathrm r}\big)\mathbf A_a. \end{aligned} \label{eq:appA_blocks}
\end{equation}
Let
\begin{align}
\big[\mathbf e_{l_1,l_2}^{\pm}\big]_{m_1} &=|a_{m_1}|\sum_{m_2=1}^{M}|a_{m_2}|s_{m_1,m_2}^{\pm}\notag\\
\big[\mathbf p_{l_1,l_2}^{\pm}\big]_{m_2} &=|a_{m_2}|\sum_{m_1=1}^{M}|a_{m_1}|s_{m_1,m_2}^{\pm}. \label{eq:appA_direct_ep}
\end{align}
Then, we can obtain
\begin{align}
\boldsymbol\nu_{l_1,l_2}^{\pm} &=\frac{2\pi}{\lambda} \begin{bmatrix} \varphi_{k,l_1}^{\mathrm r}\mathbf e_{l_1,l_2}^{\pm} -\varphi_{k,l_2}^{\mathrm r}\mathbf p_{l_1,l_2}^{\pm}\\ \vartheta_{k,l_1}^{\mathrm r}\mathbf e_{l_1,l_2}^{\pm} -\vartheta_{k,l_2}^{\mathrm r}\mathbf p_{l_1,l_2}^{\pm} \end{bmatrix}\notag\\
d_{l_1,l_2}^{\pm} &=\sum_{m_1=1}^{M}\sum_{m_2=1}^{M} |a_{m_1}a_{m_2}|c_{m_1,m_2}^{\pm}. \label{eq:appA_direct_nud}
\end{align}

Define
\begin{equation}
\begin{aligned} &(\mathbf U_i,\boldsymbol\nu_i,d_i)=\sum_{l_1,l_2=1}^{L_k^{\mathrm r}}|b_{l_1}b_{l_2}| \big(\mathbf U_{l_1,l_2}^{\pm},\boldsymbol\nu_{l_1,l_2}^{\pm},d_{l_1,l_2}^{\pm}\big). \end{aligned} \label{eq:appA_direct_aggregate}
\end{equation}
where the $+$ coefficients are used for $i=1$, while the $-$ coefficients are used for $i=2$. Furthermore, by applying \eqref{eq:appA_upper_cos} and \eqref{eq:appA_lower_cos} to \eqref{eq:appA_direct_expansion}, we obtain
\begin{align}
q(\mathbf r_k)&\leq\mathbf r_k^{\mathrm T}\mathbf U_1\mathbf r_k -\mathbf r_k^{\mathrm T}\boldsymbol\nu_1+d_1 \label{eq:appA_direct_upper}\\
q(\mathbf r_k)&\geq\mathbf r_k^{\mathrm T}\mathbf U_2\mathbf r_k -\mathbf r_k^{\mathrm T}\boldsymbol\nu_2+d_2. \label{eq:appA_direct_lower}
\end{align}
Since $\mathbf U_1\succeq\mathbf0$ and $\mathbf U_2=-\mathbf U_1\preceq\mathbf0$, the upper and lower bounds are convex and concave, respectively. Equality holds at $\mathbf r_k=\mathbf r_{k,0}$, which completes the proof.

\section{Proof of Proposition~\texorpdfstring{\ref{Prop:SCA2}}{2}}\label{App:B}
Similar to Appendix~\ref{App:A}, we have
\begin{align}
\label{Eq:RealSCA} & \mathrm{Re}( \mathbf{b}^{\mathrm{H}}\mathbf{G}_{k}(\mathbf{T})\mathbf{a}) = \mathrm{Re}\bigg( \sum_{l=1}^{L^{\mathrm{t}}_{k}}\sum_{n=1}^{N}b^{*}_{l}a_{n}e^{j\frac{2\pi}{\lambda}\rho^{\mathrm{t}}_{k,l}(\mathbf{t}_{n})} \bigg) \notag \\
& = \sum_{l=1}^{L^{\mathrm{t}}_{k}}\sum_{n=1}^{N}\vert b_{l} \vert \vert a_{n} \vert\cos\big(\frac{2\pi}{\lambda}\rho^{\mathrm{t}}_{k,l}(\mathbf{t}_{n}) + \angle b^{*}_{l}a_{n} \big)\notag \\
& = \sum_{l=1}^{L^{\mathrm{t}}_{k}}\vert b_{l} \vert \sum_{n=1}^{N}\vert a_{n} \vert\cos\big(\overline{f}_{l}(\mathbf{t}_{n})\big)
\end{align}
where $\overline f_l(\mathbf t_n) =2\pi\rho_l(\mathbf t_n)/\lambda+\angle(b_l^*a_n)$. Applying \eqref{eq:appA_upper_cos} at the current feasible position $\mathbf t_{n,0}$ yields
\begin{equation}
\re(\mathbf b^{\mathrm H}\mathbf G_k(\mathbf T)\mathbf a) \leq\sum_{l=1}^{L_k^{\mathrm t}}|b_l| \big(\mathbf t^{\mathrm T}\mathbf\Omega_l\mathbf t -\mathbf t^{\mathrm T}\bm\mu_l+c_l\big)
\end{equation}
where $\overline{\mathbf a}=[|a_1|,\ldots,|a_N|]^{\mathrm T}$ and
\begin{align*}
\mathbf\Omega_l &=\frac{2\pi^2}{\lambda^2} \begin{bmatrix} \varphi_l^{\mathrm t,2}\operatorname{diag}(\overline{\mathbf a}) &\varphi_l^{\mathrm t}\vartheta_l^{\mathrm t}\operatorname{diag}(\overline{\mathbf a})\\ \varphi_l^{\mathrm t}\vartheta_l^{\mathrm t}\operatorname{diag}(\overline{\mathbf a}) &\vartheta_l^{\mathrm t,2}\operatorname{diag}(\overline{\mathbf a}) \end{bmatrix}\\
\bm\mu_l &=\frac{2\pi}{\lambda} [\varphi_l\overline{\mathbf e}_l^{\mathrm T}, \vartheta_l\overline{\mathbf e}_l^{\mathrm T}]^{\mathrm T}, \overline{\mathbf e}_l^{\mathrm T}=[|a_1|\overline e_{1,l},\ldots,|a_N|\overline e_{N,l}],\\
\overline e_{n,l} &=\sin\big(\overline f_l(\mathbf t_{n,0})\big) +\rho_{k,l}^{\mathrm t}(\mathbf t_{n,0})\\
\overline d_l &=\sum_{n=1}^{N}|a_n|\cos\big(\overline f_l(\mathbf t_{n,0})\big)\\
&\qquad\qquad+\sin\big(\overline f_l(\mathbf t_{n,0})\big) \rho_l(\mathbf t_{n,0}) +\tfrac12\rho_{k,l}^{\mathrm t,2}(\mathbf t_{n,0}).
\end{align*}
Then,
\begin{equation}
\re(\mathbf b^{\mathrm H}\mathbf G_k(\mathbf T)\mathbf a) \leq\mathbf t^{\mathrm T}\mathbf\Omega\mathbf t -\mathbf t^{\mathrm T}\bm\mu+c
\end{equation}
where $\mathbf\Omega=\sum_{l=1}^{L_k^{\mathrm t}}|b_l|\mathbf\Omega_l$, $\bm\mu=\sum_{l=1}^{L_k^{\mathrm t}}|b_l|\bm\mu_l$, and $c=\sum_{l=1}^{L_k^{\mathrm t}}|b_l|c_l$. Since $\mathbf\Omega\succeq\mathbf0$, Proposition~\ref{Prop:SCA2} follows. This completes the proof.

\section{Proof of Proposition~\texorpdfstring{\ref{Prop:MMSE-SOMP}}{3}}\label{App:C}
The MSE of the $k$-th user is given by
\begin{align}
e_k&=|1-\mathbf v_k^{\mathrm H}\mathbf C_k^{\mathrm H}\widehat{\mathbf H}_k\mathbf B\mathbf w_k|^2+\omega\sigma_k^2\|\mathbf v_k\|^2\notag\\
&\quad+\mathbf v_k^{\mathrm H}\mathbf C_k^{\mathrm H}\widehat{\mathbf H}_k\mathbf B\mathbf W_{\neg k}\mathbf W_{\neg k}^{\mathrm H}\mathbf B^{\mathrm H}\widehat{\mathbf H}_k^{\mathrm H}\mathbf C_k\mathbf v_k\notag\\
&\quad+\sum_{r=1}^{R}\sum_{i=1}^{Q}\alpha_r^2\mathbf v_k^{\mathrm H}\mathbf C_k^{\mathrm H}\overline{\mathbf g}_{r,k}^{(i)}\overline{\mathbf g}_{r,k}^{(i),\mathrm H}\mathbf C_k\mathbf v_k \label{eq:appC_mse_expansion}
\end{align}
where $\mathbf W_{\neg k}$ represents the precoder for all users except user $k$. Then, using a compact representation, \eqref{eq:appC_mse_expansion} can be rewritten as
\begin{align}
e_k=\left\| \begin{bmatrix}1\\\mathbf0\end{bmatrix} -\begin{bmatrix}(\widehat{\mathbf H}_k\mathbf B\mathbf w_k)^{\mathrm H}\\ \mathbf\Xi_k^{\mathrm H}\end{bmatrix}\mathbf C_k\mathbf v_k \right\|^2+\zeta_1\|\mathbf v_k\|^2
\end{align}
where $\mathbf\Xi_k=[\widehat{\mathbf H}_k\mathbf B\mathbf W_{\neg k}, \alpha_1\mathbf J_{1,k},\ldots,\alpha_R\mathbf J_{R,k}]$. Equivalently, we can obtain
\begin{equation}
e_k=\|\mathbf e_k-\mathbf E_k\mathbf C_k\mathbf v_k\|^2 +\zeta_1\|\mathbf v_k\|^2 \label{eq:appC_rls}
\end{equation}
where
\begin{equation}
\mathbf E_k=\big[\widehat{\mathbf H}_k\mathbf B\mathbf W, \alpha_1\mathbf J_{1,k},\ldots,\alpha_R\mathbf J_{R,k}\big]^{\mathrm H}. \label{Eq:Theta}
\end{equation}
This completes the proof.
\section{Proof of Proposition~\texorpdfstring{\ref{Prop:MM-SOMP}}{4}}\label{App:D}
By the definition of $\mathbf M$, we have
\begin{equation}
\sum_{k=1}^{K}\big[{\mathbf F}^{(i_{d})}_{2,k}\big]_{1,1} =\tr(\mathbf M). \label{eq:appD_constant}
\end{equation}
Similarly, the definitions of $\mathbf M$, $\mathbf N$, and $\zeta_2$ yield
\begin{align}
&\sum_{k=1}^{K}2\re\big(\big[{\mathbf F}^{(i_{d})}_{2,k}\big]_{1,2}\widehat{\mathbf h}_k^{\mathrm H}\mathbf B\mathbf w_k\big)\notag\\
&=\sum_{k=1}^{K}\big(\big[{\mathbf F}^{(i_{d})}_{2,k}\big]_{1,2}\widehat{\mathbf h}_k^{\mathrm H}\mathbf B\mathbf w_k
+\mathbf w_k^{\mathrm H}\mathbf B^{\mathrm H}\widehat{\mathbf h}_k\big[{\mathbf F}^{(i_{d})}_{2,k}\big]_{1,2}^{\mathrm H}\big)\notag\\
&=-\tr\big(\mathbf N\widehat{\mathbf H}\mathbf B\mathbf W\big)
-\tr\big(\mathbf W^{\mathrm H}\mathbf B^{\mathrm H}\widehat{\mathbf H}^{\mathrm H}\mathbf N^{\mathrm H}\big)\notag\\
&=-\tr\big(\mathbf M^{-\frac12}\mathbf N\widehat{\mathbf H}\mathbf B\mathbf W\mathbf M^{\frac12}\big)\notag\\
&\qquad\qquad\qquad-\tr\big(\mathbf M^{\frac12}\mathbf W^{\mathrm H}\mathbf B^{\mathrm H}\widehat{\mathbf H}^{\mathrm H}\mathbf N^{\mathrm H}\mathbf M^{-\frac12}\big) \label{eq:appD_linear}
\end{align}
and
\begin{align}
&\sum_{k=1}^{K}\tr\big(\big[{\mathbf F}^{(i_{d})}_{2,k}\big]_{2,2}\big(\widehat{\mathbf h}_k^{\mathrm H}\mathbf B\mathbf W\mathbf W^{\mathrm H}\mathbf B^{\mathrm H}\widehat{\mathbf h}_k+\omega\widehat{\sigma}_{2,k}^{2}\big)\big)\notag\\
&=\sum_{k=1}^{K}\big[{\mathbf F}^{(i_{d})}_{2,k}\big]_{2,2}\|\widehat{\mathbf h}_k^{\mathrm H}\mathbf B\mathbf W\|^2\!+\!\left(\sum_{k=1}^{K}\big[{\mathbf F}^{(i_{d})}_{2,k}\big]_{2,2}\tfrac{\widehat{\sigma}_{2,k}^{2}}{P_{\mathrm{max}}}\right)\|\mathbf W\|_F^2\notag\\
&=\tr\big(\mathbf M^{-1}\mathbf N\widehat{\mathbf H}\mathbf B\mathbf W\mathbf W^{\mathrm H}\mathbf B^{\mathrm H}\widehat{\mathbf H}^{\mathrm H}\mathbf N^{\mathrm H}\big)+\zeta_2\|\mathbf W\|_F^2. \label{Eq:SOMP}
\end{align}
where we have used the identity
\begin{equation}
\big[{\mathbf F}^{(i_{d})}_{2,k}\big]_{2,2} =\big[{\mathbf F}^{(i_{d})}_{2,k}\big]_{1,1}^{-1} \big[{\mathbf F}^{(i_{d})}_{2,k}\big]_{2,1} \big[{\mathbf F}^{(i_{d})}_{2,k}\big]_{1,2}. \label{eq:appD_block_identity}
\end{equation}
Combining \eqref{eq:appD_constant}, \eqref{eq:appD_linear}, and \eqref{Eq:SOMP} yields Proposition~\ref{Prop:MM-SOMP}. This completes the proof.
\bibliographystyle{IEEEtran} 
\bibliography{IEEEabrv,ref}

\end{document}